\documentclass[10pt]{article}
\date{19.2.2024} 
\usepackage{amsfonts,amssymb,xcolor,amscd,latexsym,xspace,enumerate}
\usepackage[pagebackref, colorlinks=true,linkcolor=blue,citecolor=red]{hyperref}
\usepackage{hyperref}
\usepackage{pdfpages}
\input{liemacs10.sty} 

\newcommand{\sH}{{\sf H}}

\renewcommand\phi{\varphi}

\newcommand{\bb}{\mathbf{b}}
\newcommand{\bd}{\mathbf{d}}

\renewcommand{\mlabel}{\label}

\begin{document} 

\title{Covariant  projective representations of
  Hilbert--Lie groups \\  
  } 

\author{Karl-Hermann Neeb\footnote{Friedrich-Alexander-Universit\"at Erlangen-N\"urnberg, Erlangen, Germany; neeb@math.fau.de} 
\  and Francesco G. Russo\footnote{University of Cape Town,
    Cape Town, South Africa \endgraf \hspace{0.3cm}and  University of the Western Cape,    Bellville, South Africa; francescog.russo@yahoo.com}}

\maketitle
\begin{abstract}
  Hilbert--Lie groups are Lie groups whose Lie algebra is a
    real Hilbert space whose scalar product is invariant under the
    adjoint action. These infinite-dimensional Lie groups are the
    closest relatives to compact Lie groups. In this paper we study
    unitary representations of these groups from various perspectives.

    First, we address norm-continuous, also called bounded, representations.
    These are well-known for simple groups, but 
    the general picture is more complicated. Our first main result
    is a characterization of the discrete decomposability of all bounded
    representations in terms of boundedness of the set of coroots.
    We also show that bounded representations of type II and III exist
    if the set of coroots is unbounded. 

    Second, we use covariance with respect to a one-parameter group of
    automorphisms to implement some regularity. 
    Here we develop some perturbation theory based on half Lie groups
    that reduces matters to 
    the case where a ``maximal torus'' is fixed, so that compatible weight
    decompositions can be studied. 

    Third, we extend the context to projective representations which are
    covariant for a one-parameter group of automorphisms. Here important
    families of representations arise from ``bounded extremal weights'',
    and for these the corresponding central extensions can be determined
    explicitly, together with all one-parameter groups for which a covariant
    extension exists.\\[1mm]

\end{abstract}

\newpage

\tableofcontents 

\newpage

\section{Introduction}
\mlabel{sec:0}

We call a real Lie algebra $\fg$ endowed with an invariant 
positive definite form turning it into a Hilbert space a 
{\it Hilbert--Lie algebra}. Accordingly, we call a 
Lie group $G$ whose Lie algebra $\fg$ is a Hilbert--Lie algebra 
a {\it Hilbert--Lie group}. Hilbert--Lie groups are the closest 
infinite-dimensional relatives of compact Lie groups.
To understand their unitary representation theory thus 
is a natural building block in the unitary representation
theory of infinite-dimensional Lie groups.
The present work addresses unitary representations of
Hilbert--Lie groups from a rather general perspective.

The simple infinite-dimensional Hilbert--Lie algebras
(all closed ideals are trivial) 
are of the form $\fg = \fu_2(\sH)$ (skew-adjoint Hilbert--Schmidt 
operators), where $\sH$ is an infinite-dimensional real, complex 
or quaternionic Hilbert space (\cite{Sch60}). 
This generalizes the well-known simple compact Lie algebras 
$\fo(n)$, $\fu(n)$ and $\sp(n)$.  
If $B_2(\sH)$ denotes the ideal of Hilbert--Schmidt 
operators on $\sH$, then one obtains the associated Lie groups  
$\U_2(\sH) = \U(\sH) \cap (\1 + B_2(\sH))$, consisting of unitary
operators that are Hilbert--Schmidt perturbations of the identity. 
Hilbert--Lie algebras are orthogonal direct sums of their 
center and their simple ideals (\cite[\S 1.2, Thm.~1]{Sch60}).
We call a Hilbert--Lie algebra {\it semisimple} 
if its center is trivial. 

Section~\ref{sec:2} recalls the cornerstones of the
structure of Hilbert--Lie algebras and the classification of the
simple ones, due to Schue \cite{Sch60, Sch61}. 
It thus provides a solid
foundation for the representation theory of Hilbert--Lie groups. 
To understand one-parameter groups of automorphisms
of Hilbert--Lie groups,
it is instrumental to know that they preserve all simple ideals
(Lemma~\ref{lem:pres-ideals}).  We also recall the root decomposition of
simple infinite-dimensional Hilbert--Lie algebras
because the boundedness of the set
$\Delta^\vee$ of coroots plays a key role for bounded representations
(Theorem~\ref{thm:5.7}). 

The main results of this paper are contained in Sections~\ref{sec:bounded},
\ref{sec:4} and \ref{sec:5}. Each of these sections is devoted to a specific class
of representations: the bounded (=norm continuous) ones
(Section~\ref{sec:bounded}),
those that are covariant for one-parameter groups
of automorphisms (Section~\ref{sec:4}), 
and projective  representations with
suitable semiboundedness properties (Section~\ref{sec:5}). 


We start in Section~\ref{sec:bounded} with bounded
representations. 
For compact groups, it is well-known that 
their unitary representations are direct sums of irreducible 
ones and that the irreducible representations are finite-dimensional. 
For Hilbert--Lie groups the situation is more complicated.
However, the {\it bounded} representations 
are the most well-behaved. They are defined by 
$\pi$ to be continuous with respect to the norm topology on~$\U(\cH)$. 
By basic Lie theory, they correspond to
Lie algebra homomorphisms $\g \to \fu(\cH)$ for a complex
Hilbert space~$\cH$. A key result is
Theorem~\ref{thm:5.7}, which leads to a complete classification, 
provided the set $\Delta^\vee$ of coroots is bounded, 
which is the case for simple Hilbert--Lie algebras.
Then bounded representations are direct sums of irreducible ones, and these
can be classified in terms of extremal weights. 
This should in particular be compared with the corresponding
result for unitary representations of compact groups,
where the boundedness condition is not needed to arrive at the
same conclusion (\cite{HM06, HiNe12}).

For the group $\U_2(\sH)$, where $\sH \cong \ell^2(J,\C)$ is complex, 
the classification of bounded unitary representations  
was obtained in \cite{Ne98}, 
and the irreducible ones where parametrized by
``highest weights'' $\lambda = (\lambda_j)_{j \in J}$. They all arise
as irreducible subrepresentations of finite tensor products 
(\cite[Thm.~3.17]{Ne14b}, \cite{Ki73}, \cite{Ol78}, \cite{Bo80}).
In \cite{Ne12} these results have 
been extended to the real and quaternionic case. 
This shows that the bounded representations 
of simple Hilbert--Lie groups resemble the representations 
of compact groups. 
This analogy can even be extended to a 
Borel--Weil theory for coadjoint orbits of the group $\U_2(\sH)$, 
which only produces the bounded irreducible representations 
(\cite{Bo80}; see also \cite{Ne04}). In \cite[Thms.~5.2, 5.5]{Bo80} Boyer 
also constructs factorial representation 
of $\U_2(\sH)$ of type II$_1$ and type III, respectively. Here we
refer to von Neumann algebras of types I, II and III  as per  \cite{Sa71}).
This shows that, unlike the situation 
for compact groups,  the (unbounded) unitary representation 
theory of $\U_2(\sH)$ is not type I. 

In the light of Boyer's result, it is interesting that,
for semisimple Hilbert--Lie algebras with unbounded coroot systems,
bounded representations of type II and III also exist
(Theorem~\ref{thm:5.7b}).
As a consequence, a classification of the
irreducible ones is not a well-posed problem 
(\cite{Gli60, Sa67}). This provides a rather clear picture of the bounded
unitary representations of a general Hilbert--Lie group.

To go beyond bounded representations, we restrict our attention
to the simple Hilbert--Lie algebra 
$\g = \fu_2(\sH)$, where $\sH$ is a separable infinite-dimensional
complex Hilbert space, and the corresponding group~$G = \U_2(\sH)$. In Section~\ref{sec:4} we ask
which restrictions it imposes
for a unitary representation $(\pi, \cH)$ of~$G$ 
to be extendable to a representation of 
$G \rtimes_\alpha \R$ for a one-parameter group
$(\alpha_t)_{t \in \R}$ of automorphisms of~$G$.
For any such $\alpha$, there exists a selfadjoint operator
$H$ on $\sH$ with
\[ \alpha_t(g) = e^{itH} g e^{-itH} \quad \mbox{ for }\quad
  t\in \R, g \in G\]
(Theorem~\ref{thm:aut-grp}). 
Therefore the problem is non-trivial if and only if
$H$ is not a Hilbert--Schmidt operator.

In Subsection~\ref{subsec:WvN}  we prove a rather general 
Perturbation Theorem (Theorem~\ref{thm:4.4})
that can in particular be used to reduce questions concerning
covariance of unitary representations with respect to some $\alpha$
to  the case where $H$ is diagonalizable
(Corollary~\ref{cor:5.3}). This observation is based on 
a Weyl--von Neumann decomposition $H = H_d + H_{\rm HS}$, where
$H_d$ is diagonalizable and $H_{\rm HS}$ is Hilbert--Schmidt.
We may  assume henceforth that $H$ is diagonalizable,
which implies that it fixes a maximal abelian subalgebra
$\ft \subeq \g$ pointwise. We write
  \[ T := \exp \ft \subeq G \]
  for the corresponding ``maximal torus'', an abelian connected
  Banach--Lie group. 
Then we show that any factorial unitary representation
$(\pi,\cH)$ that is a {\it $T$-weight representation},
i.e., has a basis of $T$-eigenvectors,
extends to $G \rtimes_\alpha \R$ (Proposition~\ref{prop:4.11}).
As these results apply equally well to projective representations,
we briefly discuss in Subsection~\ref{subsec:centext} some techniques
for the identification of the corresponding cocycles. An important
feature of Hilbert--Lie algebras is that the invariant scalar product
provides a bijection between (bounded) derivations and $2$-cocycles. 

For representations of $G \rtimes_\alpha \R$,  we can ask for the existence
of ground states, as 
recently studied in \cite{MN16, NR24},
but also investigate the semiboundedness of the extended representation,
or conditions
for $U_t := \hat\pi(e,t)$ to have spectrum which is bounded from below.
To this end, we need to understand the structure of the subgroup 
$G^0 = \Fix(\alpha)$ of $\alpha$-fixed points in $G$,
which is described in Subsection~\ref{subsec:galpha}.
Moreover, we need to understand which representations
of $G^0$ can possibly arise on the minimal $U$-eigenspace. 
As $G$ contains a dense subgroup $G_{\rm alg} \cong \U_\infty(\C)$,
which is a direct
limit of compact Lie groups, the techniques developed in
\cite{NR24} provide important necessary conditions for this to be
the case (Theorem~\ref{thm:4.20} in Subsection~\ref{subsec:4.4}).

In Subsection~\ref{rem:factorrep}, we discuss a class of 
representations which is a very special case of a 
construction that has been used in \cite{Bo80} to show that
$\U_2(\sH)$ has factor representations of type II$_1$. 
In our context, this provides representations of $\U_2(\sH)$ 
that are covariant for the full automorphism group, 
but contain no non-zero $T$-eigenvector.  


The next larger class to bounded representations
are the semibounded ones,
whose theory has been developed 
in  \cite{Ne08, Ne10b, Ne12, Ne14, NZ13}.
Since all finite-dimensional continuous unitary representations are bounded 
and most of the unitary representations appearing in physics are semibounded 
(see \cite{DA68, PS86, Ne10b, KR87}),
semiboundedness is a natural regularity
condition for unitary representations of infinite-dimensional Lie groups. 

Theorem~\ref{theoremofreconstruction} 
asserts that semibounded unitary representations of connected
Hilbert--Lie groups decompose discretely into bounded ones,
hence provide nothing new beyond what we have seen already in
Section~\ref{sec:bounded}.
This is complemented by Theorem \ref{thm:b.6}, which,
for semibounded projective representations, shows
that they can all be lifted to representations of the simply connected
covering group. Its proof involves the remarkable Bruhat--Tits Fixed
Point Theorem (see Appendix \ref{app:b}).

As a consequence, enlarging our context to projective semibounded
  representations of $G$ does not lead to  new representations,
but projective representations of the groups
  $G \rtimes_\alpha \R$ for which the spectrum of the one-parameter group
  $\pi(e,t)$ is bounded below are far more interesting.
  Here non-trivial central $\T$-extensions appear.
  Luckily, the Hilbert space structure of~$\fg$ is quite helpful to
  represent continuous Lie algebra cocycles by derivations.

As we go along, we learn that it is too restrictive
to consider $\U_2(\sH)$ as a
Lie group, in particular, when we consider
one-parameter groups $(\alpha_t)_{t \in \R}$ of automorphisms
that act only continuously and not smoothly, i.e., their
infinitesimal generator is unbounded.
This is taken into account in
Subsection \ref{subsec:4.3}, where we show that the existence of a
unitary representation of  $\U_2(\sH) \rtimes_\alpha \R$, 
whose restriction to the Lie group
$\U_2(\sH)^\infty \rtimes_\alpha \R$
is semibounded, requires the infinitesimal generator $i\bd$ of $\alpha$
to be semibounded as an operator on $\sH$ 
(Corollary~\ref{cor:5.8}). This result builds on some fine analysis
of the coadjoint action of the semidirect product Lie group
(Theorem \ref{thm:5.8}).
If $\bd$ is unbounded, then
the topological group $\U_2(\sH) \rtimes_\alpha \R$ is not a
Lie group, but it is a smooth Banach manifold and even a
half Lie group in the sense of \cite{MN18}, a fact 
that has already been exploited in the proof of
Theorem~\ref{thm:4.4}. Then
$\U_2(\sH)^\infty \rtimes_\alpha \R$ is the Lie group
of all elements for which right multiplications are also smooth,
so that the passage between these two perspectives is well
captured by the half Lie group concept. 

Subsections \ref{subsec:traceclassgroups} and
\ref{subsec:5.4} are devoted to projective representations
of the restricted unitary group $\U_{\rm res}(\sH,D)\subeq \U(\sH)$
 for an operator $D$ on a Hilbert space $\sH$,
where the ``restriction'' is expressed in terms of $gDg^{-1}-D$ being
Hilbert--Schmidt. We recall in Theorem \ref{lem:7.7} (from \cite{Ne04})
that the extremal weight representation
of the trace class group $\U_1(\sH)$, corresponding to a bounded weight,
all extend to projective representations of restricted groups,
which contain in particular the group $\U_2(\sH)$. So we obtain
a large family of projective representations of this Hilbert--Lie group. 
In Theorem~\ref{lem:7.7} we show that the restricted groups are maximal
``covariance groups'' of these representations of $\U_1(\sH)$.
These results provide a new perspective and 
relevant information on significant constructions
which arise naturally in Quantum Field Theory
(see \cite{PS86, KR87, Ne10b}).

Although we formulated several challenging problems in the
main text, we also add a brief section on perspectives, where we discuss
some directions of research that target important tools to
analyze unitary representations of infinite-dimensional Lie groups.

\section{Hilbert--Lie groups and Lie algebras} 
\mlabel{sec:2}

This section presents the cornerstones of the
structure of Hilbert--Lie algebras and the classification of the
simple ones, due to Schue \cite{Sch60, Sch61}. 
We start in Subsection~\ref{subsec:2.1} with
Schue's Decomposition Theorem~\ref{thm:1.3}, providing a solid
foundation for the representation theory
of Hilbert--Lie groups. In Subsection~\ref{subsec:2.2},
we recall the automorphism groups of the simple Hilbert--Lie algebras
and show that connected groups of automorphisms preserve all simple ideals. 
The root decomposition for simple infinite-dimensional Hilbert--Lie algebras
are recalled in some detail because the length of coroots
plays a key role for bounded representations, as we shall see in
Theorem~\ref{thm:5.7}. 

\subsection{Preliminaries} 
\mlabel{subsec:2.1}

The following notions are widely used in \cite{dH72, HM06, Ne00, Sch60, Sch61}.

\begin{definition} (\cite[Definition 6.3]{HM06})  \label{def:1.1} 
 A {\it real Hilbert--Lie algebra}, or briefly a {\it  Hilbert--Lie algebra}, is a real Lie algebra  $\g$, 
 endowed with a continuous Lie bracket
$ [\cdot,\cdot] \: \mathfrak{g} \times \mathfrak{g} \to \mathfrak{g}$ 
 and a positive definite inner product $\langle \cdot , \cdot \rangle$,
 such that $\g$ is a Hilbert space and
 \[ \la [x,y],z \ra = \la x,[y,z]\ra \quad \mbox{ for } \quad x,y,z \in \g.\]
 An {\it isomorphism} between two Hilbert--Lie algebras is
 an isomorphism of Lie algebras which is a linear homeomorphism.
 If we want, in addition, that it is isometric, we shall always
 say so. 
\end{definition} 

Of course, we can also introduce complex Hilbert--Lie algebras, leaving unaltered Definition \ref{def:1.1} apart from the fact that the  Hilbert space
is complex; these are called   \textit{$L^*$--algebras}
by Schue in \cite{Sch60, Sch61}. In view of
\cite[Prop.~6.2, Thm.~6.6]{HM06}, there is no difference between
finite-dimensional Hilbert--Lie algebras and {\it compact} Lie algebras,
i.e.~real Lie algebras which are the Lie algebra of some compact Lie group.
One could also formulate Definition~\ref{def:1.1}
omitting the continuity of the Lie bracket.
Then the Closed Graph Theorem~\cite[Thm.~2.15]{Ru73}
and the Uniform Boundedness Principle~\cite[Thm.~2.5]{Ru73}
easily imply that the Lie bracket is continuous with respect
to the norm topology, see~\cite[p.~70, \S 1.4]{Sch60}.
We write
\begin{equation}
  \label{eq:cg}
  c_\g := \sup \{ \|[x,y]\| \: \|x\|, \|y\| \leq 1 \}
\end{equation}
for the {\it norm of the Lie bracket}. Note that, replacing the norm
by $\|\cdot\|' := \lambda \|\cdot\|$, this constant changes to
$c_\g' = \lambda^{-1} c_\g$. 


\begin{defn}\label{newdefHS}
  (\cite[Ch.~1]{Sa91}) 
  Let $\sH$ be a Hilbert space over $\K \in \{\R,\C,\H\}$, i.e.,
  real, complex of quaternionic. We follow the convention that the
inner product  is linear in the second argument. 
  We write $B(\sH)$ for the space of $\K$-linear bounded operators.

  For $0 \leq X \in B(\sH)$ self-adjoint and  an orthonormal basis
  $(e_j)_{j \in J}$ for $\sH$,  the \textit{trace} of $X$ is defined by 
\begin{equation}\label{trf}
\tr(X):=\sum_{j \in J} \langle e_j, X e_j \rangle \in [0, \infty],
\end{equation}
the {\it trace norm} of $x$ is defined by $\|X\|_1 := \tr(\sqrt{X^*X})$.
We say that $X$ is of {\it trace class} if its trace norm is finite. 
Then 
\begin{equation}
  B_1(\sH) := \{X \in B(\sH)  \mid \ \|X\|_1< \infty \} \end{equation}
is a Banach space with respect to the trace norm.
It is a $*$-ideal in $B(\sH)$.
Likewise the space 
\begin{equation}
B_2(\sH) := \{X \in B(\sH)  \mid \ \tr(X^*X) < \infty \}, \end{equation} 
of {\it Hilbert--Schmidt operators} is a Banach space with respect to the
{\it Hilbert--Schmidt norm} $\|X\|_2:=\tr(X^*X)^{1/2}$.  
\end{defn}

\begin{defn} \mlabel{newdefHSbis} (\cite{Ne02})
  If $\sH$ is an arbitrary Hilbert space over $\K \in \{\R, \C, \H\}$, the {\it unitary group} of $\sH$ and its Lie algebra are denoted by 
\begin{equation} \U(\sH) := \{ g \in \GL(\sH) \mid g^{-1} = g^*\} \quad \mbox{ and } \quad 
\fu(\sH) := \{ X \in B(\sH) \mid X^* = -X \},
\end{equation} 
where $\GL(\sH)$ denotes the
group of all invertible continuous operators of $\sH$.
For $p = 1,2$, we shall also need the Banach--Lie groups
\[ \GL_p(\sH) := \GL(\sH) \cap (\1 + B_p(\sH)) \quad \mbox{ and }  \quad 
  \U_p(\sH) := \U(\sH) \cap (\1 + B_p(\sH))\]
with Lie algebra
\[ \gl_p(\sH) = (B_p(\sH), [\cdot,\cdot]) \quad \mbox{ and } \quad
  \fu_p(\sH) := \fu(\sH) \cap B_p(\sH).\] 

For $\K=\R$, we also write on the group level 
\[  \OO(\sH) := \U(\sH), \quad
  \OO_p(\sH) := \OO(\sH) \cap \GL_p(\sH), \quad \fo(\sH) := \fu(\sH)
\quad \mbox{ and } \quad \fo_p(\sH) := \fu_p(\sH).\] 
We likewise write for $\K=\H$ 
\begin{equation} \Sp(\sH) := \U(\sH), \quad \Sp_p(\sH) := \U_p(\sH), \quad  
\fsp(\sH) := \fu(\sH) \quad \mbox{ and } \quad 
\fsp_p(\sH) := \fu_p(\sH).\end{equation} 
\end{defn}

\begin{rem}\mlabel{schatten-id}
Recall from \cite[Prop~A.I.7(v)]{Ne00} that
\begin{equation}\label{hoelder-halfbis}
  {\|XY\|}_2 \le {\|X\|} \ {\|Y\|}_2 \ \ \mbox{for} \    X \in B(\sH),
  Y \in B_2(\sH),
\end{equation} 
and  that $\|X\|_2 \le \|X\|_1$ if $X \in B_1(\sH)$ by
\cite[Prop.~A.I.9(ii)]{Ne00};
furthermore
\begin{equation}
  \label{eq:tracenormesti}
  \|XY\|_1 \le {\|X\|} \ {\|Y\|}_1 \quad \mbox{ for } \quad X \in B(\sH),
  Y \in B_1(\sH)
\end{equation}
by \cite[Prop.~A.I.9(i)]{Ne00}.
\end{rem}

\begin{defn} (\cite{Sch60, Sch61})
  We call a Hilbert--Lie algebra $\g$ {\it simple}
  if $\{0\}$ and $\g$ are the only closed ideals of $\g$.
  It is called {\it semisimple} if its center $\fz(\g)$ is trivial.
\end{defn}

We now turn to the main examples of Hilbert--Lie algebras.

\begin{ex}
  \mlabel{nex:1.3}
\nin (a) If $\sH$ is a $\K$-Hilbert space with $\K \in \{\R,\C,\H\}$ 
and $(e_j)_{j \in J}$ is a  $\K$-orthonormal basis, then 
\[ \fu_2(\sH):= \{ x \in B_2(\sH) : x^* = - x\} \] 
is a Hilbert--Lie algebra  
with respect to the real bilinear scalar product 
\[ \la X,Y \ra = \Re(\tr_\K(X^*Y))
= \sum_{j \in J} \Re\la X e_j, Y e_j\ra.\]  

\nin (b) If $\sH$ is infinite-dimensional, then this Hilbert--Lie algebra
is simple. If $\sH \cong \K^n$ is finite-dimensional, then
it may have a center:
\begin{itemize}
\item For $\K = \C$, we have the decomposition
  \[ \fu_2(\sH) = \fu_n(\C) = \su_n(\C) \oplus \R i \1,
    \quad \mbox{ where } \quad
    \su_n(\C) = \{ X \in \fu_n(\C) \: \tr X = 0 \}.\]
\item For $\K =  \R$, we have
  $\fu_2(\sH) = \so_n(\R)$, which is simple for $n \geq 5$ or $n = 3$,
  but 
  \[ \so_4(\R) \cong \so_3(\R) \oplus \so_3(\R) \quad\mbox{ and } \quad
    \so_2(\R) = \fz(\so_2(\R)) \cong \R. \]
\item For $\K = \H$, the Lie algebra
  $\fu_2(\sH) \cong \fu_n(\H)$ is always simple.     
\end{itemize}
\end{ex}

The following theorem reduces many questions on 
Hilbert--Lie algebras to simple ones. Its main point
is the existence of simple ideals when $\g$ is not abelian.
Note that the norm on a simple Lie algebra
is only unique up to a positive factor and this
has a substantial influence on infinite direct sums
(see Lemma~\ref{lem:infsum} for details). We 
put \ $\hat{}$\ \  on $\oplus$ whenever it refers to a complete
  direct sum, such as for Hilbert space direct sums.

\begin{thm} {\rm(Schue's Decomposition Theorem; \cite[Ths.\ 1,2]{Sch60})} 
  \mlabel{thm:1.3} 
  Every  Hilbert--Lie algebra $\g$ is an orthogonal Hilbert
  direct sum of its center 
$\fz(\g)$ and simple ideals $(\g_j)_{j \in J}$:
\begin{equation}
  \label{eq:idealdecomp}
  \g = \z(\g) \oplus \hat\bigoplus_{j \in J} \g_j.
\end{equation}
Each simple infinite-dimensional 
Hilbert--Lie algebra is non-isometrically isomorphic to 
$\fu_2(\sH)$ for an infinite-dimensional real, complex or 
quaternionic Hilbert space $\sH$ and
the finite-dimensional simple Hilbert--Lie algebras
are the simple compact Lie algebras. 
\end{thm}

It is easy to see that every ideal of $\g$ is adapted to the
decomposition \eqref{eq:idealdecomp}. 

\begin{rem}
Note that every infinite-dimensional $\K$-Hilbert space
$\sH$ is isomorphic to some $\ell^2(J,\K)$ for an infinite
set $J$, whose cardinality determines the isomorphism class.
Accordingly, two  simple Hilbert--Lie algebras 
$\fu_2(\ell^2(J_1,\K_1))$ and $\fu_2(\ell^2(J_2,\K_2))$
are isomorphic if and only if $|J_1| = |J_2|$ and
$\K_1 = \K_2$.
\end{rem}

\subsection{Automorphism groups}
\mlabel{subsec:2.2}

In this short  subsection we collect some results
on automorphism groups of Hilbert--Lie algebras.
We describe the automorphisms of the simple factors and show that
connected groups of automorphisms preserve all simple factors
(Lemma~\ref{lem:pres-ideals}). The latter result will be used 
in Proposition~\ref{prop:4.10}
for a characterization of one-parameter groups of unitary automorphisms
of the Lie algebra of Hilbert--Schmidt operators.

For a complex Hilbert space $\sH$, we write 
$\AU(\sH)$ for the group of unitary (or antiunitary) isometries 
of $\sH$ and 
\[ \PAU(\sH) := \AU(\sH)/\T\1 \cong \PU(\sH) \rtimes \{\id,\sigma\}, \] 
denotes the projective unitary (or antiunitary) isometries of $\sH$,
where  $\sigma$ is an anticonjugation of $\sH$. 

We recall the structure of the automorphism groups
of  simple infinite-dimensional Hilbert--Lie algebras: 

\begin{thm}
  \mlabel{thm:aut-grp} {\rm(\cite[Thm.~1.15]{Ne14})} 
The automorphism groups of the simple infinite-dimensional 
Hilbert--Lie algebras are given by 
\[  \Aut(\fu_2(\sH)) \cong \PAU(\sH) \] 
for a complex Hilbert space $\sH$,
and for the real and quaternionic case we have 
\[  \Aut(\fo_2(\sH)) \cong \OO(\sH)/\{\pm \1\} \quad \mbox{ and } \quad 
  \Aut(\fsp_2(\sH)) \cong \Sp(\sH)/\{\pm \1\}.\]
\end{thm} 

It follows in particular that the group $\Aut(\fu_2(\sH))$
always has $2$ connected components, that $\Aut(\sp_2(\sH))$ is always
connected, and that  $\Aut(\fo_2(\sH))$ is connected if and only if
$\dim \sH$ is not even and finite (\cite{Ne02}). 
Automorphisms of finite order of simple
  infinite-dimensional Hilbert--Lie algebras
  are discussed in \cite[App.~D]{Ne14}, with the purpose to study
  related twisted loop groups with Hilbert target groups. 
The group $\Aut(\g)$ of the topological Lie algebra
automorphisms of the Hilbert--Lie algebra~$\g$ 
preserves the center $\fz(\g)$ and permutes the
simple ideals $\g_j$ in the decomposition~\eqref{eq:idealdecomp}.
The following lemma draws an important
conclusion for the action of connected groups on $\g$.
It applies in particular to continuous $\R$-actions by automorphisms.

\begin{lem} \mlabel{lem:pres-ideals}
If a connected topological group 
  acts on a Hilbert--Lie algebra $\g$ 
  by isometric automorphism and continuous orbit maps,
  then it preserves all  simple ideals of $\g$.   
\end{lem}

\begin{prf} Let $A$ be a connected topological group acting on
  $\g$ by automorphisms and with continuous orbit maps.
  Further, let   $(\g_j)_{j \in J}$ denote the simple ideals of $\g$.
    For any unit vector $x_j \in \g_j$, the function
  \[ f \: A \to \R, \quad f(\gamma) = \la x_j, \gamma.x_j \ra \]
  is continuous with $f(e) = 1$. Hence there exists an identity
  neighborhood $U \subeq A$ with $f(\gamma) > 0$ for $\gamma \in U$.
  Then $\gamma \in U$ implies that $\gamma.x_j$ is not orthogonal to $x_j$,
  hence cannot be contained in any other simple ideal.
  As $\gamma$ permutes the simple ideals, it follows that
  $\gamma(\g_j) = \g_j$. This shows that the stabilizer
  \[ A_j := \{ \gamma \in A \: \gamma(\g_j) = \g_j \} \]
  is an open subgroup of $A$, thus equal to $A$ by connectedness.
  We conclude that $A$ leaves $\g_j$ invariant.   
\end{prf}

\subsection{Root decomposition} 
\mlabel{subsec:2.3}

In this section we review some basic facts
concerning  real and complex root space decompositions of Hilbert--Lie algebras
(\cite{Bo80, MN16, MN17, Ne12, Ne98, NeSt01, Sch60, Sch61}).
In the context of Hilbert--Lie algebras,
the role of Cartan subalgebras from the finite-dimensional
theory (\cite{HiNe12}) is taken over by the
maximal abelian subalgebras; recall that for compact Lie algebras
these are precisely the Cartan subalgebras.
The existence of maximal abelian subalgebras follows immediately
from Zorn's Lemma, 
but, unlike the finite-dimensional context, they are not all conjugate
(\cite{NeSt01}).
However, Schue's work always implies the existence of a corresponding
root space decomposition in the topological sense
(\cite{Sch61}).

\begin{defn} \mlabel{def:basic2} (\cite[\S IX.1]{Ne00}) 
Let $\g$ be a real Hilbert--Lie algebra and $\g_\C$ be its complexification. 
For $x,y \in \g$, we put $(x+iy)^* := -x + i y$,
which defines an antilinear antiautomorphism of the complex
Lie algebra $\g_\C$,  so that 
\begin{equation} 
\g = \{ x \in \g_\C \mid x^* = -x\}.
\end{equation} 
Let $\ft \subeq \g$ be a maximal abelian subalgebra 
and $\ft_\C\subeq \g_\C$ be its complexification. 
For a linear functional $\alpha \in \ft_\C^*$ (the algebraic dual space), 
\begin{equation} \g_\C^\alpha = \{ x \in \g_\C \mid (\forall h \in \ft_\C)\ 
[h,x]= \alpha(h)x\}
\end{equation} 
is called  the corresponding {\it root space}, and 
\begin{equation} \Delta := \Delta(\g_\C,\ft_\C) :=  \{ \alpha \in \ft_\C^* \setminus \{0\} \mid \g_\C^\alpha
\not= \{0\}\}
\end{equation} 
is called the {\it root system} of the pair $(\g_\C,\ft_\C)$.
We then have $\g_\C^0 = \ft_\C$ because $\ft$ is maximal abelian,
and
\[ [\g_\C^\alpha, \g_\C^\beta] \subeq \g_\C^{\alpha+\beta}
  \quad \mbox{ for } \quad \alpha,\beta\in \ft_\C^* \]
follows from
the fact that $\ad\ft$ consists of derivations; in particular 
$[\g_\C^\alpha, \g_\C^{-\alpha}] \subeq \ft_\C$. 
\end{defn}  

Let $\g$ be a Hilbert--Lie algebra and 
$\ft \subeq \g$ be a maximal abelian subalgebra. 
According to \cite{Sch61}, $\ft_\C \subeq \g_\C$ defines an 
orthogonal root space decomposition 
$$ \g_\C = \ft_\C \oplus \hat\bigoplus_{\alpha \in \Delta} \g_\C^\alpha $$
which is a Hilbert space direct sum (indicated by the hat over $\oplus$).
As for compact Lie algebras, all root spaces 
$\g_\C^\alpha = \C x_\alpha$ are one-dimensional 
and $\alpha([x_\alpha, x_\alpha^*]) > 0$ for $0 \not= x_\alpha \in \g_\C^\alpha$. 
For every $\alpha \in \Delta$ we have 
$\alpha(\ft) \subeq i \R$ and therefore
$x \in \g_\C^\alpha$ implies $x^* \in \g_\C^{-\alpha}$.

\begin{defn} (Coroots and Weyl group, \cite[\S 2.2]{Sch61}) \mlabel{coroot} 
    From $\alpha([x_\alpha, x_\alpha^*]) > 0$ it follows that
  there exists a unique element 
\[ \alpha^\vee \in \ft_\C \cap [\g_\C^\alpha, \g_\C^{-\alpha}]
  \quad \mbox{ satisfying } \quad \alpha(\alpha^\vee) = 2.\]
It is called the {\it coroot of $\alpha$} 
and  $\Delta^\vee:=\{\alpha^\vee \mid \alpha \in \Delta\}$
is called the \textit{coroot system}. 

The {\it Weyl group} $\cW = \cW(\g_\C,\ft_\C)
\subeq \GL(\ft_\C)$ is the subgroup generated by 
all reflections 
\begin{equation}
  \label{eq:ref1}
 r_\alpha(x) := x - \alpha(x) \alpha^\vee \quad \mbox{ for } \quad 
\alpha \in \Delta_c. 
\end{equation}
It acts on the dual space $\ft_\C^*$ by the adjoint maps 
\begin{equation}
  \label{eq:ref2}
  r_\alpha^*(\beta) := \beta - \beta(\alpha^\vee) \alpha
  \quad \mbox{ for } \quad \beta \in \ft_\C^*,2.10
\end{equation}
and this action preserves the root system $\Delta$. 
\end{defn}

In the finite-dimensional case $\cW$ is finite (\cite[Prop.~VII.2.10]{Ne00}),
but this is no longer true in the infinite-dimensional case
(\cite[Rem.~V.2]{Ne98}).

If $\g$ is semisimple, then 
 \[ \ft_{\rm alg}  := i \Spann_\R(\Delta^\vee)\subeq \ft \] 
is a dense subspace of $\ft$. We write 
\[ \la x + iy, x' + i y' \ra := \la x,x'\ra + \la y,y' \ra
  + i(\la x, y'\ra - \la y, x' \ra) \] 
for the  {\bf unique hermitian extension} 
 of the scalar product on $\g$ to $\g_\C$
(linear in the second argument). It 
satisfies 
\[ \la [x,y],z \ra = \la y, [x^*,z]\ra \quad \mbox{ for } \quad 
  x,y,z \in \g_\C.\]

We denote by $\ft'$, resp., $\ft_\C'$, 
the space of continuous linear functionals on $\ft$, resp., $\ft_\C$.
For $\alpha \in \ft_\C'$, we write 
$\alpha^\sharp \in \ft_\C$ for the uniquely determined element satisfying 
\begin{equation}
  \label{eq:sharp}
\alpha(h) = \la \alpha^\sharp,h\ra \quad \mbox{ for } \quad h \in \ft_\C.
\end{equation}
The bijection 
$\sharp \: \ft_\C' \to \ft_\C$ is antilinear, and 
\[ \alpha(\beta^\sharp)
  = \la \alpha^\sharp, \beta^\sharp \ra
  = \la \beta^\sharp, \alpha^\sharp \ra
  = \beta(\alpha^\sharp) \quad 
\mbox{ for  } \quad \alpha,\beta \in i \ft'.\]
For $h \in \ft_\C$ and $x_\alpha \in \g_\C^\alpha$, the relation 
\[ \la h, [x_\alpha^*, x_\alpha]\ra
=  \la [x_\alpha, h], x_\alpha\ra
=  -\alpha(h) \la x_\alpha, x_\alpha\ra
=  -\la h, \|x_\alpha\|^2\alpha^\sharp \ra \] 
then leads to 
\begin{equation}
  \label{eq:brel}
[x_\alpha, x_\alpha^*] = \| x_\alpha\|^2 \alpha^\sharp.
\end{equation}
From $\alpha(\alpha^\vee) = 2$, $\alpha^\sharp \in \C \alpha^\vee$ 
and $\alpha(\alpha^\sharp) = \|\alpha^\sharp\|^2 = \|\alpha\|^2$ 
it follows that 
\begin{equation}
  \label{eq:normrel}
  \alpha^\vee = \frac{2}{\|\alpha\|^2} \alpha^\sharp
  \quad \mbox{ and therefore } \quad
  \|\alpha^\vee\| = \frac{2}{\|\alpha\|}.
\end{equation}

The preceding discussion provides a rather direct way to determine
  the number $c_\g$ in terms of the root decomposition:

\begin{lem} \mlabel{lem:cg} We have
  \[ c_\g = \sup \{ \|\alpha\| \: \alpha \in \Delta\}
    = \frac{2}{\inf \{ \|\alpha^\vee\| \: \alpha \in \Delta \}}.\]
\end{lem}

\begin{prf} Let $0 \not= x_\alpha \in \g_\C^\alpha$.
  Then $y_\alpha := x_\alpha - x_\alpha^* \in \g$ with orthogonal summands, so
  that 
  \[ \|y_\alpha\|^2 = 2 \|x_\alpha\|^2 = \|i(x_\alpha + x_\alpha^*)\|,\]
  and
  \[ \|[i\alpha^\vee, y_\alpha]\|
    = 2\| i(x_\alpha + x_\alpha^*)\|
    = 2 \|y_\alpha\| \leq c_\g \|\alpha^\vee\| \|y_\alpha\|,\]
  so that $c_\g \geq \frac{2}{\|\alpha^\vee\|} = \|\alpha\|.$ 

  Conversely, for $h \in \ft$, the root decomposition shows that
  \[ \|\ad h\|  = \sup \{ |\alpha(h)\| \: \alpha \in \Delta \}
    \leq \|h\| \sup \{ \|\alpha\| \: \alpha \in \Delta \}.\]
  Since every element in the dense subspace
  \[ \g_{\rm alg} := \ft_{\rm alg} 
+ \sum_{\alpha \in \Delta} \{ x_\alpha - x_\alpha^* \: x_\alpha \in \g_\C^\alpha\}\] 
  is conjugate under an inner automorphism to an element of $\ft$,
  and since conjugate elements have the same norm, it follows by continuity that
  $ \|\ad x\| \leq  \|x\| \sup \{ \|\alpha\| \: \alpha \in \Delta \}$
  for $x \in \g,$ 
  and hence that $c_\g \leq \sup \{ \|\alpha\| \: \alpha \in \Delta \}$.
\end{prf}

After these preparations, we can now introduce weights;
see also \cite{MN17, HiNe12, Ne00, Ne98} for a similar formalization.

\begin{defn} (Bounded and continuous weights, \cite[\S\S I.11-14]{Ne98})
  A linear functional $\lambda \: (\ft_{\rm alg})_\C \to \C$ 
is called a {\it weight} or an {\it integral weight} if $\lambda(\Delta^\vee) \subeq \Z$. 
A weight $\lambda$ is said to be {\it bounded} if 
$\lambda(\Delta^\vee)$ is a bounded subset of $\Z$. 
We write $\cP\subeq i \ft_{\rm alg}^*$ for the group of weights and 
$\cP_b \subeq \cP$ for the subgroup of bounded weights. 

Since not all bounded weights are continuous on $\ft_{\rm alg}$
with respect to the subspace topology inherited from $\ft$
(Example~\ref{ex:d.1a}), 
we cannot identify $\cP_b$ with a subset of the topological dual
space~$\ft_\C'$. Accordingly, we write 
$\cP_c \subeq \cP_b$ for the subgroup of continuous bounded weights.
\end{defn}

We now describe the relevant root data for the three concrete types of 
Hilbert Lie algebras $\fu_2(\sH)$, where $\sH$ is a 
Hilbert space over $\K \in \{\R,\C,\H\}$. 

\begin{ex} \mlabel{ex:d.1a} 
(Root data of unitary Lie algebras, \cite[Ex.~C.4]{Ne12})   
Let $\sH$ be a complex Hilbert space with 
orthonormal basis $(e_j)_{j \in J}$. 
Let $\g := \fu_2(\sH)$ and let $\ft \subeq \g$ be the 
subalgebra of all diagonal operators with respect to the $e_j$. 
Then $\ft$ is a maximal abelian subalgebra. The set of 
roots of $\g_\C= \gl_2(\sH)$ with respect to $\ft_\C$  is given by the root 
system 
\[ \Delta = \{ \eps_j - \eps_k \: j\not= k \in J \} =: A_J.\]
Here the operator 
$E_{jk} e_m := \delta_{km} e_j$ is a $\ft_\C$-eigenvector 
in $\gl_2(\sH)$ generating the corresponding eigenspace 
and $\eps_j(\diag(h_k)_{k \in J}) = h_j$. 
From $E_{jk}^* = E_{kj}$ it follows that 
\[   (\eps_j - \eps_k)^\vee= E_{jj} - E_{kk} 
  = [E_{jk}, E_{kj}] = [E_{jk}, E_{jk}^*],\]
which leads with \eqref{eq:normrel} to
\begin{equation}
  \label{eq:coroot-unit}
  \|(\eps_j - \eps_k)^\vee\| = \sqrt{2}
  \quad \mbox{ and thus } \quad 
  \|\eps_j - \eps_k\| = \sqrt{2}. 
\end{equation}
With Lemma~\ref{lem:cg} we now obtain $c_\g = \sqrt{2}.$ 

\nin {\bf Weights:} We have 
\[ \ft_{\rm alg} = i\Spann \{ E_{jj} - E_{kk} \: j\not=k \},  \] 
and each linear functional on $\ft_{\rm alg}$ can be represented by a 
function $\lambda \: J \to \C, j \mapsto \lambda_j$ via 
\[ \lambda(E_{jj} - E_{kk}) = \lambda_j - \lambda_k.\] 
It is a weight if and only if all differences 
$\lambda_j - \lambda_k$ are integral, and in this case 
it can be represented by a $\Z$-valued function 
$\lambda\: J \to \Z$. In this case its boundedness is equivalent to 
$\lambda(J)$ being finite and it is continuous if and only if
its support $\supp(\lambda)$ is finite. 

\nin {\bf Weyl group}:
The Weyl group $\cW$ is isomorphic to the group $S_{(J)}$ of finite 
permutations of $J$, acting in the canonical  way on $\ft_\C$. 
It is generated by the reflections $r_{jk} := r_{\eps_j - \eps_k}$ 
corresponding to the transpositions of $j$ and~$k$. 
The Weyl group acts transitively on the set of roots and, in particular, 
all roots have the same length with respect to the dual norm on~$\ft_\C'$. 

This generalizes the root decomposition of
the reductive Lie algebra $\fu_n(\C)_\C = \gl_n(\C)$ with respect to the
Cartan subalgebra of diagonal matrices. It matches the general
context if $J = \{1,\ldots,n\}$ is finite. 
\end{ex}

\begin{rem} \mlabel{rem:real-struc}
In many situations it is convenient to describe 
real Hilbert spaces as pairs 
$(\sH,\tau)$, where $\sH$ is a complex Hilbert space 
and $\tau \: \sH\to \sH$ is a {\it conjugation}, i.e., 
an antilinear isometry with $\tau^2 = \id_\sH$. 

A quaternionic Hilbert space $\sH$ can be considered as a 
complex Hilbert space $\sH^\C$ (the underlying complex Hilbert 
space), endowed with an {\it anticonjugation} $\tau$, i.e., 
$\tau$ is an antilinear isometry with $\tau^2 = -\1$. 
\end{rem}

\begin{ex} \mlabel{ex:d.1c} 
(Root data of orthogonal Lie algebras, \cite[Ex.~C.6]{Ne12})  
Let $\sH$ be a real 
Hilbert space 
and $\g := \fo_2(\sH)$ be the corresponding Hilbert--Lie 
algebra. Let $\ft \subeq \g$ be a maximal abelian subalgebra. 
The fact that $\ft$ is maximal 
abelian implies that the common kernel
\begin{equation}
  \label{eq:cht}
  \sH^\ft := \bigcap_{x \in \ft} \ker x
\end{equation}
is at most 
one-dimensional. Since $\ft$ consists of compact skew-symmetric operators, 
the Lie algebra $\ft_\C$ is simultaneously
diagonalizable on the complexification 
$\sH_\C$, which implies that the space $(\sH^\ft)^\bot$
carries  an isometric complex 
structure $I$ commuting with $\ft$ and there exists an orthonormal 
subset $(e_j)_{j \in J}$ of $(\sH^\ft)^\bot$ such that 
the set $\{ e_j, Ie_j\: j \in J\}$ is a real
orthonormal basis of $(\sH^\ft)^\bot$ 
and the planes $\R e_j + \R I e_j$ are $\ft$-invariant. 
If $\sH^{\ft}$ is non-zero, we write $e_{j_0}$ for a unit vector 
in this space and put $f_{j_0} := e_{j_0}$.
For $j_0 \not= j \in J$, we  put 
\[ f_j := \frac{1}{\sqrt 2}(e_j - iI e_j) \quad \mbox{ and } 
\quad f_{-j} := \frac{1}{\sqrt 2}(e_j + iI e_j).\] 
Then $(f_j)_{j \in J}$ form a complex orthonormal basis of $\sH_\C$ consisting 
of $\ft$-eigenvectors. For the complex bilinear extension 
$\beta \: \sH_\C \times \sH_\C \to \C$ of the scalar product on $\sH$, we 
have 
\[ \beta(f_{j_0}, f_{j_0}) = 1, \quad \beta(f_j, f_j) = 
\beta(f_{-j}, f_{-j}) =0 \quad \mbox{ and } \quad  \beta(f_j, f_{-j}) = 1.\] 

We conclude that $\ft_\C$ coincides with the space of those elements 
in $\fo_2(\sH_\C,\beta) = \g_\C$ which are diagonal with the orthonormal basis consisting 
of the $f_{\pm j}$. Hence $\ft_\C \cong \ell^2(J,\C)$
(as a Hilbert space), where 
$x \in \ft_\C$ corresponds to the element $(x_j)_{j \in J} \in \ell^2(J,\C)$ 
defined by $x f_j = x_j f_j$, $j \in J$.
For $x_j \in i \R$ we then have
$x \in \ft$ with
\[ xe_j = \frac{1}{\sqrt 2} x(f_j + f_{-j})
      = \frac{x_j}{\sqrt 2} (f_j - f_{-j})
      = -i x_j I e_j \quad \mbox{ and } \quad 
      x I e_j = i x_j e_j.\]
    Therefore
    \begin{equation}
      \label{eq:norm-form}
 \|x\|_2^2 = \sum_{j \in J} \|x e_j\|^2 + \|xIe_j\|^2
 = 2 \sum_{j \in J} |x_j|^2.
    \end{equation}

Writing $\eps_j(x) := x_j$, we see that 
$\{ \pm \eps_j : j \in J\}$, together with $\eps_{j_0}$ if 
$\sH^\ft \not=\{0\}$, is the set of $\ft_\C$-weights of $\sH_\C$. 
Accordingly, the set of 
roots of $\g_\C$ with respect to $\ft_\C$  is given by 
\[ \Delta = \{ \pm \eps_j \pm \eps_k \: j \not= k, j,k\in J \} =: D_J 
\quad \mbox{ if } \quad \sH^\ft =\{0\}, \] 
and 
\[ \Delta = \{ \pm \eps_j \pm \eps_k \: j \not= k, j,k\in J \} 
\cup \{ \pm\eps_j \: j \in J\} =: B_J 
\quad \mbox{ if } \quad 
\sH^\ft \not=\{0\}.\] 
If $\sH$ is finite-dimensional and $\dim \sH = 2n$ is even, 
then $\sH^\ft = \{0\}$, and we obtain a root system of type~$D_n$. If $\dim \sH = 2n + 1$ is odd, then $\Delta$ is of type $B_n$
(\cite[Ex.~6.3.10]{HiNe12}).

We put $E_j := E_{jj} - E_{-j,-j}$ for $j \not= j_0$ and $E_{j_0} := E_{j_0, j_0}$.
If $\Delta$ is of type $D_J$, then the coroots are of the form 
$\pm E_j \pm E_k, j\not=k$, 
and 
for $\Delta$ of type $B_J$ the coroots are of the form 
\begin{equation}
  \label{eq:rootvec-bj}
  (\pm \eps_j \pm \eps_k)^\vee = \pm E_j \pm E_k
  \quad \mbox{ for } \quad j\not=k \quad \mbox{ and }\quad
\pm \eps_j^\vee = \pm 2 E_j, j \in J. 
\end{equation}
With \eqref{eq:norm-form} we now obtain
\[ \|(\eps_j \pm \eps_k)^\vee\| = 2 \quad \mbox{ and } \quad
  \| \eps_j^\vee\| = 2 \sqrt{2},\]
which leads with Lemma~\ref{lem:cg} to $c_\g = 1.$

\nin {\bf Weights:} Here 
\[ \ft_{\rm alg} = i\Spann \{ E_j \: j\in J \}  \] 
and each linear functional on $\ft_{\rm alg}$ can be represented by a 
function $\lambda \: J \to \C$ via 
$\lambda(E_j) = \lambda_j.$ 
For both types of root systems, 
the weights are characterized by the condition 
\[ \lambda(J) \subeq \Z \quad \mbox{ or } \quad 
\lambda(J) \subeq \shalf + \Z\]  
(see \cite[Props.~V.3, VII.2, VII.3]{Ne98} for details).

\nin {\bf Weyl group}:
For $B_J$ we obtain the same Weyl group 
$\{\pm 1\}^{(J)} \rtimes S_{(J)}$ as for $C_J$. 
For $D_J$ the reflection $r_{\eps_j + \eps_k}$ changes the sign of the 
$j$- and the $k$-component, so that the 
Weyl group $\cW$ is isomorphic to the group $\{\pm 1 \}^{(J)}_{\rm even}
\rtimes S_{(J)}$,
where $\{\pm 1 \}^{(J)}_{\rm even}$ is the subgroup of all those 
elements with an even number of entries equal to $-1$. 
For $D_J$ the Weyl group acts transitively on the set of roots and, 
in particular, all roots have the same length. 
For $B_J$ we have two $\cW$-orbits,  the roots $\pm \eps_j$ are short 
and the roots $\pm \eps_j \pm \eps_k$, $j \not=k$, are long. 
\end{ex}

\begin{ex} \mlabel{ex:d.1b} 
(Root data of symplectic Lie algebras, \cite[Ex.~C.5]{Ne12}) 
For a complex Hilbert space $\sH$ with a conjugation $\tau$, 
we consider the quaternionic Hilbert space $\sH_\H := \sH^2$, where the 
quaternionic structure is defined by the anticonjugation 
$\tilde\tau(v,w) := (\tau w, -\tau v)$ (cf.\ Remark~\ref{rem:real-struc}). 
Then 
\[ \g := \sp_2(\sH_\H) = \fu_2(\sH_\H) 
= \{ x \in \fu(\sH^2) \: \tilde\tau x = x \tilde\tau \}\]  
and 
\[ \sp_2(\sH_\H)_\C = \left\{ \pmat{ A & B \cr C & -A^\top \cr} 
\in B(\sH^2) \: B^\top = B,  C^\top  = C\right\}. \] 
Let $(e_j)_{j \in J}$ be an orthonormal basis of $\sH$ and 
$\ft \subeq \g \subeq \fu_2(\sH^2)$ be the 
subalgebra of all diagonal operators with respect to the basis elements  
$(e_j,0)$ and $(0,e_k)$ of $\sH^2$. Then $\ft$ 
is maximal abelian in $\g$. 
Moreover, $\ft_\C \cong \ell^2(J,\C)$, 
consists of diagonal operators of the form 
$h = \diag((h_j), (-h_j))$ with
\[  \|h\|_2^2
  = \sum_{j \in J} \|he_j\|^2 
= \sum_{j \in J} |h_j|^2.\]  
The set of 
roots of $\g_\C$ with respect to $\ft_\C$  is given by 
\[ \Delta = \{ \pm 2 \eps_j, \pm (\eps_j \pm \eps_k) \: j \not= k, j,k
\in J \} =: C_J,\] 
where $\eps_j(h)  = h_j$. If we write 
$E_j \in \ft_\C$ for the element defined by 
$\eps_k(E_j) = \delta_{jk}$, then the coroots are given by 
\begin{equation}
  \label{eq:CJcoroot}
 (\eps_j \pm  \eps_k)^\vee= E_j \pm  E_k 
\quad \mbox{ for } \quad j \not=k 
\quad \mbox{ and } \quad 
(2\eps_j)^\vee= E_j
\end{equation}
and
\[  \|(\eps_j \pm  \eps_k)^\vee\| = \sqrt2, \qquad 
\|(2\eps_j)^\vee\| = 1,\]
which leads with Lemma~\ref{lem:cg} to $c_\g = 2.$ 
Here the roots $\eps_j - \eps_k$ correspond to block diagonal operators, 
 the roots $\eps_j + \eps_k$ to strictly upper triangular operators, 
and the roots $-\eps_j - \eps_k$ to strictly lower triangular operators. 

\nin {\bf Weights:} In particular, we have 
\[ \ft_{\rm alg} = i\Spann \{ E_j  \: j\in J \}  \] 
and we can represent functionals on $\ft_{\rm alg}$ by 
functions $\lambda \: J \to \C$, so that 
integrality is equivalent to $\lambda(J) \subeq \Z$ 
(\cite[Sect.~VI]{Ne98}). 

\nin {\bf Weyl group}:
The Weyl group $\cW$
is isomorphic to the group $\{\pm 1\}^{(J)} \rtimes S_{(J)}$, 
where the factor on the left acts by  finite 
sign changes on $\ell^2(J,\R)$. The 
reflection $r_{\eps_j - \eps_k}$ acts as a transposition and 
the reflection $r_{2\eps_j}$ changes the sign of the $j$th component. 
The Weyl group has two orbits in $C_J$, the short roots form a root system 
of type $D_J$ and the second orbit is the set 
$\{ \pm 2 \eps_j \: j \in J\}$ of long roots. 
\end{ex}

From the preceding discussion of the Hilbert--Lie algebras
$\fu_2(\sH)$ for Hilbert spaces over $\K = \R,\C, \H$, we
obtain the following nice relation:
\begin{prop} The norm $c_\g$ of the commutator bracket in $\g = \fu_2(\sH)$
  over $\K \in \R,\C,\H$, satisfies 
\begin{equation}
  \label{eq:cgu2}
  c_\g^2 = \dim_\R \K \in \{1,2,4\}.
\end{equation}
  \end{prop}

\section{Bounded representations of Hilbert--Lie groups}
\mlabel{sec:bounded}

In this section we describe the bounded unitary representations
of connected Hilbert--Lie groups $G$,
which by basic Lie theory, correspond to
Lie algebra homomorphisms $\pi \: \g \to \fu(\cH)$ for a complex
Hilbert space~$\cH$. The main result is
Theorem~\ref{thm:5.7}, which leads to a complete classification, 
provided the set $\Delta^\vee$ of coroots is bounded, 
which is the case for simple Hilbert--Lie algebras.
These representations are direct sums of irreducible ones, which 
can be classified in terms of extremal weights. 
For semisimple algebras with unbounded coroot systems,
representations of type II and III occur (Theorem~\ref{thm:5.7b}).

\subsection{Direct sums of Hilbert--Lie algebras} 

According to Schue's Theorem, 
any Hilbert--Lie algebra $\g$ is an orthogonal direct sum 
$\fz(\g) \oplus \hat\oplus_{j \in J} \g_j$, where the $\g_j$ are simple 
ideals. On each simple ideal $\g_j$ the norm is only determined up to a 
positive factor. Changing these factors leads to Hilbert--Lie algebras 
whose representation theory behaves very differently. 

\begin{defn} (\cite{Ka90}, \cite{PS86}) 
  The inner product of a simple Hilbert--Lie algebra 
$\g$ is said to be {\it normalized} if 
the long roots $\alpha$ satisfy $\|\alpha\|^2 = 2$. 
The inner product of a general Hilbert--Lie algebra 
$\g$ is said to be {\it normalized} if it is normalized on each 
simple ideal. 
\end{defn}

\begin{lem} \mlabel{lem:infsum} Let $(\g_j)_{j \in J}$ be a family of 
  Hilbert--Lie algebras for which there exist $c,C > 0$ 
  with
  \begin{equation}
    \label{eq:c-sandwich}
    c \leq c_{\g_j} \leq C \quad \mbox{ for all } \quad j \in J.
  \end{equation}
For each function $w \: J \to (0,\infty)$, we define a
real inner product on the algebraic direct sum 
$\bigoplus_{j \in J} \g_j$ via 
\begin{equation}
  \label{eq:wprod}
  \la (x_j), (y_j) \ra := \sum_{j \in J} w_j \la x_j, y_j \ra.
\end{equation}
Then the Hilbert completion $\g$ with respect to the corresponding norm
is a Lie algebra if and only  if $\inf (w_j)_{j \in J} > 0$. 
\end{lem}

\begin{prf} The Hilbert completion $\g$ of $\g^0 := \bigoplus_{j \in J} \g_j$
with respect to \eqref{eq:wprod}  is a Hilbert--Lie algebra
  if and only if the bracket on $\g^0$ is continuous, we have 
to show that 
\[ c_\g := \sup \{ \|[x,y]\| \: x,y \in \g^0, \|x\|, \|y\| \leq 1 \} < \infty\] 
is equivalent to $\inf (w_j)_{j \in J} > 0$. 
In view of 
\[ c_\g \geq \sup \{ \|[x,y]\| \: x, y \in \g_j, 
\sqrt{w_j} \|x_j\|, \sqrt{w_j} \|y_j\| \leq 1 \} 
=  \frac{c_{\g_j}}{w_j}, \] 
the condition $\inf (w_j)_{j \in J} > 0$ is necessary for $c_\g < \infty$. 

Suppose,  conversely, that $\inf (w_j)_{j \in J} > 0$. Let 
$x,y \in \g^0$ with 
$\|x\|^2, \|y\|^2 \leq 1.$ 
Then 
\begin{align*}
 \|[x,y]\|^2 
&= \sum_j w_j \|[x_j,y_j]\|^2 
\leq \sum_j c_j^2 w_j \|x_j\|^2 \|y_j\|^2
                \leq \sup_{j \in J} \|y_j\|^2 \cdot \sup\{c_j\: j \in J\}^2
                \sum_j  w_j \|x_j\|^2 \\
  &\leq   \sup_{j \in J} \|y_j\|^2 \cdot \sup\{c_j\: j \in J\}^2
    \leq \frac{\sup\{c_j\: j \in J\}^2}{\inf \{ w_j \: j \in J\}}
    < \infty.
\end{align*}
This proves the lemma with
\[ c_\g     \leq \frac{\sup\{c_j\: j \in J\}^2}{\inf \{ w_j \: j \in J\}}.
\qedhere\]
\end{prf}

\begin{prob} Suppose that $J = \N$ and that the sequence
    $(w_n)_{n \in \N}$ of weights satisfies $w_n \to 0$, so that the
    above construction leads to a Hilbert space $\g$ to which
    the Lie bracket does not extend continuously.
    This is a typical situation one finds in the theory of half-Lie groups
    (cf.\ \cite{MN18, BHM23}). Is there a half-Lie group
    $G$ with tangent space $T_e(G) = \g$? A natural candidate would be
    \[ G := \Big\{  (g_n) \in\prod_{n \in\N} G_n \:
      \sum_{n = 1}^\infty w_j d_{G_n}(g_n,e)^2 < \infty \Big\},\]
    where $d_{G_n}$ is the natural biinvariant
    metric on the simple Hilbert--Lie group $G_n$. Note however, that
    exponential charts do not work because the exponential function is
    singular on any $0$-neighborhood in $\g$.
\end{prob}

In view of Lemma~\ref{lem:cg}, the normalization condition is equivalent
to $c_\g = \sqrt 2$. This numerical restriction has some interesting consequences: for instance, if the assumption
  \eqref{eq:c-sandwich}  on the constants $c_{\g_j}$
  in  the preceding lemma is satisfied,
  they can always be normalized to the value $\sqrt{2}$ 
  without changing the weighted direct sum Hilbert--Lie algebra $\g$.
Then the norms 
of the coroots of the simple Lie algebras $\g_j$
are bounded from below and above and we have
\begin{equation}
  \label{eq:normhatal}
 \|\alpha^\vee\|_\g^2 = w_j \|\alpha^\vee\|^2_{\g_j} \quad
 \mbox{ for } \quad \alpha \in \Delta(\g_j,\ft_j).
\end{equation}
Therefore $\inf (w_j)_{j \in J} > 0$ if and only if
the norms of the coroots of $\g$ are bounded from below.

\begin{rem}
    \mlabel{rem:3.8}
    From \eqref{eq:normhatal} we further derive that, in
    the context of Lemma~\ref{lem:infsum},
    the set $\Delta_\g^\vee$ of coroots is bounded if and only if
    the family $(w_j)_{j \in J}$ is bounded. As it is also assumed
    to satisfy $\inf (w_j)_{j \in J} > 0$, this implies that the 
    Hilbert--Lie algebra $\g$ simply is the Hilbert direct sum
    $\hat\oplus_{j \in J} \g_j$, provided the norms on the simple ideals
    $\g_j$ are normalized. So we may as well put all weights
    $w_j$ equal to~$1$.
\end{rem}

  \subsection{Decomposing bounded representations} 
    
We recall some standard notation, which is used in operator theory and in the theory of von Neumann algebras; in particular it is used in the next proposition. 
For a subset $\cS \subseteq B(\cH)$, \textit{ the commutant of $\cS$} is
\[ \cS'=\{A \in B(\cH) \mid (\forall B \in \cS) \ AB=BA\}.\]
A {\it von Neumann algebra} can also be
  defined via its commutant, in fact it is a $*$-invariant subalgebra $\cM \subeq B(\cH)$ satisfying $\cM = \cM''$.
For a $*$-invariant set $\cS$ of operators, $\cS''$ is the von Neumann algebra generated by $\cS$. A von Neumann algebra $\cM$
is called a {\it factor} if its center $\cZ(\cM) := \cM \cap \cM'$ coincides with $\C \1.$
A unitary representation $(\pi,\cH)$ of a group $G$ is called {\it factorial} if
$\pi(G)''$ is a factor.
For the concept of factors of type I, II, III, we refer to \cite[\S 2.2]{Sa71}.

To deal with factor representations
of semisimple Hilbert--Lie groups, the following proposition
is an extremely useful tool. For finite products
and type I representations, it reduces the problem to the individual
factor groups. 

\begin{prop}\mlabel{prop:4.2.7b}
{\rm(Factoriality for product groups)}   If 
$(\pi,\cH)$ is a factor representation of the product 
group $G = G_1 \times G_2$, then the following assertions
hold:
\begin{itemize}
\item[\rm(a)] $\pi\res_{G_1}$ and $\pi\res_{G_2}$ are factor representations. 
\item[\rm(b)] If one of these is of type $I$, i.e.,
contains irreducible subrepresentations, then there exist 
factor representations $(\pi_j, \cH_j)$ of $G_j, j =1,2$, 
with $\pi \cong \pi_1 \otimes \pi_2.$
If, in addition, $\pi$ is irreducible, then both representations
$\pi_1$ and $\pi_2$ are irreducible. 
\end{itemize}
\end{prop}

\begin{prf} (a) We identify $G_1$ and $G_2$ with the corresponding 
  subgroups of $G$. 
Since  $\pi(G_1) \subeq \pi(G_2)'$, we have 
\[ \cZ(\pi(G_1)'') = \pi(G_1)' \cap \pi(G_1)''  
  \subeq \pi(G_1)' \cap \pi(G_2)' \cap \pi(G)'' =
  \pi(G)' \cap \pi(G)'' = \cZ(\pi(G)'') = \C \1. \] 
So $\pi\res_{G_1}$ is a factor representation. 
A similar argument shows that $\pi\res_{G_2}$ is a factor representation. 

\nin (b) If $\pi\res_{G_1}$ is of type $I$, then the Structure
Theorem for type I factors implies that
$\cH \cong \cH_1 \hat\otimes \cH_2$
with $\pi(G_1)'' = B(\cH_1) \otimes \1$
and $\pi(G_2) \subeq \pi(G_1)' = \1 \otimes B(\cH_2)$
(\cite[Thm.~2.3.3]{Sa71}). 
Hence there exist representations
$(\pi_j, \cH_j)_{j = 1,2}$ of $G_j$, such that
\[ \pi(g_1) = \pi_1(g_1) \otimes \1 \quad \mbox{ and } \quad
  \pi(g_2) = \1 \otimes \pi_2(g_2) \quad \mbox{ for } \quad
  g_j \in G_j,\]
so that $\pi = \pi_1 \otimes \pi_2$.
If $\pi_1$ or $\pi_2$ is not irreducible, then $\pi$
cannot be irreducible. Therefore the irreducibility of $\pi$
implies that $\pi_1$ and $\pi_2$ are both irreducible.
\end{prf}

The following theorem extends a well-known result
on finite-dimensional tori $ \T^n$ to
a natural class of infinite-dimensional abelian
Banach--Lie groups that we encounter 
in particular as  identity components of centers of Hilbert--Lie groups. 
Far a discussion of concrete examples of such groups,
we refer to Remark~\ref{rem:toroid} later on.

\begin{thm} {\rm(Discrete decomposition of bounded representations---the abelian
    case)}
  \mlabel{thm:z-discrete}
  Let $Z$ be a connected abelian Banach--Lie group
  and $ \hat Z := \Hom(Z,\T)$ be its character group.
  Then the following assertions hold:
  \begin{itemize}
  \item[\rm(a)] If $\hat Z$ is countable, then all
    norm-continuous unitary representations are
    direct sums of eigenspaces, respectively,
    of irreducible subrepresentations. 
  \item[\rm(b)] If $Z$ is separable and all its norm-continuous
    unitary representations are direct sums of eigenspaces, then
    $\hat Z$ is  countable. 
  \end{itemize}
\end{thm}

\begin{prf} Let $\fz := \L(Z)$ be the Lie algebra of $Z$, so that
  $Z \cong \fz/\Gamma$, where $\Gamma \subeq \fz$ is a discrete
  subgroup (cf.\ \cite{GN, Ne06}). Then
  \begin{equation}
    \label{eq:hatz}
    \hat Z \cong \{ \alpha \in \fz' \: \alpha(\Gamma) \subeq \Z \},
  \end{equation}
  where we associate to $\alpha \in \fz'$ the character
  $\chi_\alpha(z + \Gamma) := e^{2\pi i \alpha(z)}$.

\nin (a) Let $(\pi, \cH)$ be a norm continuous
  unitary representation of $Z$.
  By \cite[Thm.~4.1]{Ne08}, this representation is obtained
  from a spectral measure on the character group $\hat Z$.
  As $\hat Z$ is countable, this spectral measure is atomic
  (see \cite{Ru73} concerning the notion of atomic measures), and thus
  \[ \cH = \hat\bigoplus_{\chi \in \hat Z} \cH_\chi
    \quad \mbox{ with } \quad
    \cH_\chi = \{ v \in \cH \: (\forall z \in Z)\
    \pi(z) v = \chi(z)v\}.\]

 \nin (b) We argue by contradiction. As in \eqref{eq:hatz} above, we identify
  $\hat Z$ with a weak-$*$-closed subgroup of $\fz'$.
  If this group is uncountable, there exists an $R > 0$, such that
  \[  \hat Z_R := \{ \alpha \in \hat Z \: \|\alpha \| \leq R\} \]
  is uncountable. The separability of $Z$ now implies that $\hat Z_R$,
  endowed with the weak-$*$-topology, is a metrizable compact space,
  hence also separable and therefore second countable. By compactness,
  it is complete, hence a polish space. Let $\cB$ be the $\sigma$-algebra
  of Borel subsets of $\hat Z_R$. By \cite[Thm.~I.4]{Fa00},
  the measurable space $(\hat Z_R,\cB)$ is isomorphic to the unit interval
  $[0,1]$, endowed with the $\sigma$-algebra of Borel sets.
  In particular, there exists a probability measure $\mu$ on $\hat Z_R$ without
  atoms. We consider the Hilbert space $\cH := L^2(\hat Z_R, \mu)$, on
  which
  \[ (\pi(z)f)(\alpha) := \chi_\alpha(z) f(\alpha) \]
  defines a unitary representation of $Z$. This representation
  is norm-continuous because $\hat Z_R$ is bounded in~$\fz'$.
  If $f \in \cH$ is an eigenfunction for a character $\chi_\beta$,
  then $f$ vanishes $\mu$-almost everywhere on
  $\hat Z_R \setminus \{\chi_\beta\}$, and since $\mu$ has no atoms,
  we must have $f = 0$. Therefore $\cH$ contains no non-zero
  $Z$-eigenvectors. 
\end{prf}

Definitively we have an interesting classification for connected
Hilbert--Lie groups $G$ for which the identity component
$Z(G)_e$ of the center  $ Z(G)$ is small enough. 

\begin{thm} \mlabel{thm:5.7}
  {\rm(Discrete decomposition of bounded representations---the general case)}
  For a connected Hilbert--Lie group $G$ for which the character 
  group $\hat Z$ of $Z := Z(G)_e$ is countable and 
$\Delta^\vee \subeq \ft$ is bounded, the following assertions hold: 
\begin{itemize}
\item[\rm(i)] Every bounded 
unitary representation of $G$ is a direct sum of irreducible 
ones. 
\item[\rm(ii)] If $\g = \fz(\g) \oplus \hat\bigoplus_{j \in J} \g_j$ 
is the decomposition into center and simple ideals $\g_j$, 
then each irreducible 
bounded unitary representation is a finite tensor product of 
irreducible bounded representations of $\fz(\g)$ and of finitely many~$\g_j$. 
\item[\rm(iii)] Irreducible bounded 
unitary representations are extremal weight representations.  
\end{itemize}
\end{thm}

\begin{prf} {\bf Step 1:} From \cite[Thm.~3.14]{Ne98} and 
\cite[Thms.~D.5,6]{Ne12} we know that (i) and (iii) hold 
for infinite-dimensional simple Hilbert--Lie algebras. For finite-dimensional 
ones, it follows from the representation theory of compact groups. 
For the group $Z = Z(G)_e$, the corresponding assertion follows from
Theorem~\ref{thm:z-discrete}. 

\nin {\bf Step 2:} For a finite subset $F \subeq J$, we consider the  subalgebra 
$\g_F := \fz(\g) \oplus \sum_{j \in F} \g_j$ with simple ideals~$\g_j$. 
Since, for all simple factors, bounded representations decompose
into irreducible ones, hence are type I,  
the irreducible bounded representations of $\g_F$
are tensor products of irreducible representations of the factors 
(Proposition~\ref{prop:4.2.7b}). Hence (ii) and (iii) hold for~$\g_F$. 

To verify (i) for $\g_F$,
it suffices to show that each non-zero bounded representation 
of $\g_F$ contains an irreducible one.
Since (i) holds for each simple ideal
$\g_j$ and also for $Z$ (Theorem~\ref{thm:z-discrete}),
any bounded unitary representation of $\g_F$ 
is an orthogonal sum of $\g_j$-isotypical 
components and $\fz$-eigenspaces.
By induction this property is inherited by the direct sum $\g_F$
of $\fz$ and finitely many simple ideals~$\g_j$. 

We may therefore assume that the representation  
$(\pi, \cH)$ of $\g$ is isotypical for $\fz$ and
each $(\g_j)_{j \in F}$. This implies that its 
restriction to $\g_F$ is a 
tensor product of simple representations $\pi_j$ of the ideals $\g_j$,
a character $\chi$ of $\fz$, 
and a trivial representation. In particular, (i) follows for $\g_F$. 

\nin {\bf Step 3:} Let $\ft \subeq \g$ be maximal abelian and 
$\cP \subeq \ft'$ be the set of integral weights. For a finite subset 
$F \subeq J$, we define $\g_F$ as above. 
Our assumptions contain that 
\[ C := \sup \{ \|\alpha^\vee\| \: \alpha \in \Delta \} < \infty.\] 
If $\mu$ is a weight not vanishing 
on the simple ideals $(\g_j)_{j \in F}$, then there are 
mutually orthogonal roots 
$(\alpha_j)_{j \in F}$ with 
$\mu(\alpha^\vee_{j})\not=0$, hence 
$1 \leq |\mu(\alpha^\vee_{j})| \leq C \|\mu\res_{\g_{j}}\|$ and therefore 
\begin{equation}
  \label{eq:normmu}
\|\mu\|^2 \geq \sum_{j \in F} \|\mu\res_{\g_{j}}\|^2 
\geq \frac{|F|}{C^2}. 
\end{equation}
For a bounded representation $(\pi, \cH)$ of $\g$ we thus obtain 
$|F| \leq C^2 \|\pi\|^2$.

  \nin {\bf Step 4:} If $(\pi, \cH)$ is a non-trivial bounded representation
  of $\g$ and $j_1 \in J$ such that $\pi(\g_{j_1}) \not=\{0\}$,
  there exists a subrepresentation $(\pi_1, \cH_1)$ which
  is non-trivial and isotypical for $\g_{j_1}$. 
  If $j_2 \not= j_1$ is such that $\pi_1(\g_{j_2}) \not=\{0\}$,
  we find a subrepresentation $(\pi_2, \cH_2)$ of $(\pi_1, \cH_1)$,
  which is isotypical and non-trivial on $\g_{j_1}$ and $\g_{j_2}$.
  Proceeding like this, we find with Step 3 a maximal finite subset
  $F = \{ j_1, j_2, \ldots , j_N\} \subeq J$ and a subrepresentation
  $(\pi_N,\cH_N)$ of $(\pi,\cH)$ which is isotypical and non-trivial
  on $\g_{j_1}, \ldots, \g_{j_N}$. By maximality, all other ideals act
  trivially on $\cH_N$. 
For every bounded representation $(\pi,\cH)$ of $\g$, 
we thus obtain with Zorn's Lemma a direct sum decomposition
into subrepresentations which are isotypical on ideals of the form
$\g_F$, $F \subeq J$ finite, and trivial on the complementary ideal
$\g_{J \setminus F}$. This implies the assertion.
\end{prf}

From the classification of the bounded representations of $\fu_2(\sH)$ 
in \cite{Ne98, Ne04}, we now obtain: 

\begin{cor} \mlabel{cor:classif} If $\g$ is a Hilbert--Lie algebra and 
$\ft \subeq \g$ a maximal abelian subalgebra for which $\Delta^\vee$ is bounded, then 
the bounded irreducible unitary representations $(\pi, \cH)$ are parameterized 
(via extremal weights) 
by the set $\cP_c/\cW$ of $\cW$-orbits in the set of continuous weights 
\[ \cP_c = \{ \mu \in \ft' \: (\forall \alpha \in \Delta)\, 
\mu(\alpha^\vee) \in \Z\}. \] 
This parametrization is achieved by assigning to the
representation $\pi$ the set 
of extreme points in $\conv(\cP_\pi)$, where 
$\cP_\pi := \{ \mu \in \ft' \: \cH^\mu \not=\{0\}\}$ is the 
set of $\ft$-weights of $\pi$. 
\end{cor}

\begin{prf} In view of Theorem~\ref{thm:5.7}(iii), it suffices
  to consider the case where $\g$ is simple.
  As every bounded representation of the Hilbert--Lie algebra $\g$
  extends to a complex linear $*$-representation
  $\g_\C \to B(\cH)$, we have to classify holomorphic
  representations of the simply connected
  complex Lie group $G_\C$ with Lie algebra $\g_\C$.
  For $\g = \fu_2(\sH)$, $\dim \sH = \infty$, it 
  follows from \cite[Cor.~III.11, Thm.~III.14]{Ne98} 
  that these are highest weight representations.
  If $\g$ is finite-dimensional, this is the classical Cartan--Weyl Theorem
  of the highest weight.

  For the classification in terms of Weyl group orbits of integral
  weights, we refer to \cite[Thm.~III.16]{Ne04}. Here the main point
  is that the convex subset $\conv(\cP_\pi) \subeq i \ft'$ actually
  has extreme points, i.e, extremal weights.
  Then it turns out that the Weyl group
  $\cW$ acts transitively on the set of extremal weights,
  and this establishes
  the asserted correspondence. 
\end{prf}

For finite-dimensional connected Hilbert--Lie groups $G$, i.e.,
Lie groups with compact Lie algebra, 
we know that finite-dimensional (hence bounded)
unitary representations separate the points.
As $G$ is a product of a vector group and a compact group,
this follows from \cite[Cor.~ 2.27]{HM06}. 
The following remark shows that this result
does not extend to infinite-dimensional Hilbert--Lie groups.

\begin{rem} If $K$ is simply connected, then the corresponding 
bounded unitary representations do in general not separate 
the points. A typical example where this happens is the simply connected 
covering group $\tilde\U_2(\sH)$, where $\sH$ is a complex infinite-dimensional 
Hilbert space. For this group all bounded unitary representations 
vanish on the kernel of the covering map
$\tilde\U_2(\sH) \to \U_2(\sH)$, which is isomorphic to the fundamental 
group  $\pi_1(\U_2(\sH)) \cong \Z$, 
hence factors through $\U_2(\sH)$ (\cite{Ne98}). 
\end{rem}

\begin{rem} \mlabel{rem:bound-inherit} 
(a) Suppose that the set $\Delta^\vee$ of coroots is bounded for 
$\Delta = \Delta(\g,\ft)$ and that $\fl \subeq \g$ is a closed 
$\ft$-invariant subalgebra, so that $\fl_\C$ is 
adapted to the root decomposition of $\g_\C$. 
Then 
\[ \fl = (\fl \cap \ft) + (\fl \cap \oline{[\ft,\g]}) \] 
and all coroots corresponding to root spaces in $\fl_\C$ are 
contained in $i\ft'$ for $\ft' := \fl \cap \ft$. 
Therefore $\ft'$ is maximal abelian in $\fl$ and 
$\Delta' := \Delta'(\fl,\ft')$ clearly satisfies 
$(\Delta')^\vee \subeq \Delta^\vee$. In particular, 
the boundedness of $\Delta^\vee$ implies the boundedness of 
$(\Delta')^\vee$. 

\nin (b) From (a) it follows in particular that for all  $\ft$-invariant 
Hilbert subalgebras of a simple Hilbert algebra~$\g$, the 
set of coroots is bounded, so that Theorem~\ref{thm:5.7} applies. 
\end{rem}

\subsection{Representations of type II and III} 

The following theorem complements Theorem~\ref{thm:5.7} by asserting
that, if $\Delta^\vee$ is unbounded, then the bounded representation
theory of $\g$ is not of type~I. 

\begin{thm} \mlabel{thm:5.7b}
  If $\g$ is a Hilbert--Lie algebra for which
  $\Delta^\vee$ is unbounded, then it has bounded unitary 
  factor representations of type III and, if $\g$ is separable, also
  of type II. 
\end{thm}

\begin{prf} According to Theorem~\ref{thm:1.3}, we have
  the orthogonal direct sum of ideals
  \[ \g = \fz(\g) \oplus \hat{\bigoplus}_{j \in J} \g_j,\]
  where the ideals $\g_j$ are simple. Since our assumption
  only refers to the semisimple part of $\g$, we may assume that
  $\fz(\g) = \{0\}$.
    
  We consider on $\g_j$ the unique norm with $c_{\g_j} = 2$, so that
  \[ \inf \{ \|\alpha^\vee\|_{\g_j} \: \alpha \in \Delta_j\} = 1\]
 by Lemma~\ref{lem:cg}, 
  and since the possible ratios of the square lengths of roots
  are $\{1,2,3,4\}$ (\cite[Rem.~6.4.6]{HiNe12}), we then have
  \[ \sup \{ \|\alpha^\vee\|_{\g_j} \: \alpha \in \Delta_j\} \leq 2.\]
    
  We now write the norm on $\g$ as
  \[ \|x\|^2= \sum_{j \in J} w_j \|x_j\|_{\g_j}^2 \]
  and observe that 
  \[ w_j = \sup \{ \|\alpha^\vee\|_{\g}^2 \: \alpha \in \Delta_j,
\|\alpha^\vee\|_{\g_j}  = 1 \}.\]
  Thus $(w_j)_{j \in J}$ is unbounded. Hence there exists a sequence
  $(j_n)_{n \in \N}$ with $w_{j_n} \geq n^4$. As we can always compose
  representations of an ideal with the projection onto that ideal,
  we may replace $J$ by $\N$ and assume that
  $w_n \geq n^4$.

  For an irreducible unitary representations
  $(\rho_n, \cH_n)$ of $\g_n$ of  highest weight $\lambda_n$, we have
  $\|\rho_n\| = \|\lambda_n\|$ (this follows by reduction to
  $\ft$ and $\cP_{\rho_n} \subeq \conv(\cW\lambda_n)$).
  If $\g_n$ is finite-dimensional, we take $\rho_n$ to be the
  adjoint representations, and for infinite-dimensional ideals 
  $\g_n \cong \fu_2(\cH_n)$, we consider the representation
  on~$\cH_n$ for $\K = \C,\H$, and  on $\sH_{n,\C}$ for $\K = \R$.
  We thus obtains a sequence of bounded unitary representations
  $\rho_n$ by compact operators on $\cH_n$, and such that
  \[ r := \sup_{n \in \N} \|\rho_n\| < \infty. \] 

  Then the $C^*$-subalgebra $\cA_n := C^*(\rho_n(\g_n)) \subeq B(\cH_n)$
  coincides with the algebra $B_\infty(\cH_n)$ of compact operators 
  because it acts irreducibly on $\cH_n$
  (\cite[Cor.~4.1.6]{Dix77}). 
  As $\dim \cH_n \geq 2$ for every $n$, we may choose rank-two
  projections $P_n$ on $\cH_n$ and consider the
  $C^*$-tensor product
  \[ \cA := \hat{\bigotimes}_{n \in \N} (\cA_n, P_n), \]
  specified by the isometric embeddings
$\otimes_{j = 1}^n \cA_j \to \otimes_{j = 1}^{n+1} \cA_j, 
X \mapsto X \otimes P_{n+1}$
(See \cite{GrNe09} for more details on this construction).

  Then we have continuous morphisms of Lie algebra
  \[ \eta_n \: \g_n \to \cA, \quad
    \eta_n(x) = P_1 \otimes \cdots \otimes P_{n-1} \otimes
    \rho_n(x) \otimes P_{n+1} \otimes \cdots \]
  whose ranges commute pairwise and satisfy
$\|\eta_n\| = \|\rho_n\| \leq r.$ 
  For any finite sum $x = x_1 + \cdots + x_n \in \g_1 + \cdots + \g_n$,
  we then obtain for
  \[ \eta(x) := \rho_1(x_1) + \cdots + \rho_n(x_n) \]
  the estimate 
  \[ \|\eta(x)\|
    \leq r\sum_{j = 1}^n  \|x_j\|
    \leq r\sum_{j = 1}^n  w_j^{-1/2} \|x\|
    \leq r\|x\| \sum_{j = 1}^n  w_j^{-1/2}
    \leq r\|x\| \sum_{j = 1}^\infty  j^{-2}.\]
  Therefore the Lie algebra homomorphism
  $\eta \: \sum_n \g_n \to \cA$ extends to a continuous homomorphism
  {$\eta \: \g \to~\cA$}.

That 
  \[ \eta(\g_1 + \cdots + \g_n) \subeq
    \cA_1 \otimes \cdots \otimes \cA_n \cong
    B_\infty(\cH_1 \otimes \cdots \otimes \cH_n)\]
  is a $C^*$-generating subset
  follows from the irreducibility of the representation 
  $\rho_1 \otimes \cdots \otimes \rho_n$ 
  of $\g^n := \g_1 + \cdots + \g_n$
    (\cite[Cor.~4.1.6]{Dix77}). 
  This implies that $\eta(\g)$ is $C^*$-generating in $\cA$. 
Therefore every $C^*$-representation $(\pi,\cH)$ of $\cA$
  yields a bounded unitary representation
  $(\pi \circ \rho, \cH)$ of $\g$ with the same commutant.
In particular, every factor representation 
of $\cA$ defines a bounded factor representation of $\g$.
We now argue that $\cA$ has non-type I factor representations. 

In $\cA$ we consider the projection
\[ P_\infty := \otimes_{n \in \N} P_n.\] 
It defines the hereditary unital subalgebra
$\cB := P_\infty \cA P_\infty$ onto which
\[ \eps \: \cA \to \cB, \quad A \mapsto P_\infty A P_\infty \]
defines a conditional expectation, so in particular a completely positive map. 
We thus obtain representations of $\cA$
by Stinespring dilation from the completely positive maps
of the form $\omega = \pi \circ \eps$, where $(\pi,\cF)$ is a
representation of $\cB$ (\cite[Def.~3.17, Prop.~3.20]{BGN20}).
Concretely, we have 
\[ \cB \cong \bigotimes_{n \in \N} P_n \cA_n P_n
  \cong \bigotimes_{n \in \N} M_2(\C),\]
so that $\cB$ is a uniformly hyperfinite (UHF) algebra, hence not of type I
(\cite{Po67}).
Let $(\rho,\cH)$ be a factorial representation of $\cB$ which is not of type~I 
and consider the completely positive map $\omega := \rho \circ \eps$ on $\cA$
and its minimal dilation~$(\pi_\omega,\cH_\omega)$. This is a
representation of $\cA$ containing $\cH$ as a closed $\cA$-cyclic \break 
{$\cB$-invariant} subspace on which the $\cB$-representation is equivalent
to~$\rho$. 
Then \cite[Lemmas~3.14 and 3.19]{BGN20} imply that the commutant
$\rho(\cB)'$ is isomorphic to the commutant 
$\pi_\omega(\cA)'$. It follows that $\pi_{\omega}$ is  a factor 
representation of $\cA$ of the same type.

Finally, we use \cite[Thm.~4.6.4]{Sa71} to see that a
$C^*$-algebra which is not of type I has factorial representations
of type III. Further, \cite[\S 4.6, Rem.~2]{Sa71} implies that,
if it is separable, then it also has factor representations of type~II.
\end{prf}

The technique used in the preceding proof is based on the general
observations that bounded representations of Lie groups
can be identified with representations of suitable $C^*$-algebras,
so that the well-developed representation theory of $C^*$-algebras
can be used to study bounded unitary Lie group representations.
We refer to \cite{Ne17} for a more detailed outline of this technique.

Combing the preceding theorem
with Theorem~\ref{thm:5.7}, we obtain:

\begin{cor} Let $\g$ be a semisimple Hilbert--Lie algebra.
  Then all its bounded unitary representations are of type I if and only if
  $\Delta^\vee$ is bounded.     
\end{cor}

\begin{ex} In \cite{NeSe11} we have analyzed the bounded representations 
of Lie algebras of the type
\[ \su_2(\cA) \cong \cA \otimes \su_2(\C),\] where 
$\cA$ is a unital commutative Banach-$*$-algebra. For all these 
Lie algebras, bounded irreducible 
representations are finite-dimensional and finite tensor products 
of evaluation representations. If $\cA$ is not unital, 
this is no longer the case, as the following example shows: 

For $J = \N$, $\g_n = \su_2(\C)$  and 
$w_n := n^4$ we consider the direct sum Hilbert--Lie algebra 
\[ \g := \ell^2(\N,w, \su_2(\C)) \quad \mbox{ with } \quad 
\|x\|^2 
= \sum_n w_n \|x_n\|_2^2
= \sum_n n^4 \|x_n\|_2^2.\] 
Then $x \in \g$ implies that 
$\|x_n\|_2 \leq \frac{\|x\|}{n^2}$. 
Note that $\g$ is a Hilbert--Lie algebra by 
Lemma~\ref{lem:infsum}. 
This Lie algebra is of the form
$\g = \su_2(\cA)$, where $\cA = \ell^2(\N,w,\C)$ is 
a commutative Banach $*$-algebra with respect to the pointwise product,  and
contained in $\ell^1(\N,\C)$. 
The coroots of $\g$ are of the form 
$\alpha^\vee_n = (\delta_{jn} \alpha^\vee)$, so that the sequence 
$\|\alpha^\vee_n\|^2 = 2 w_n$ is unbounded (Theorem~\ref{thm:5.7}). 
\end{ex} 

\begin{rem} If $\sH$ is separable infinite-dimensional,
  then Theorem~\ref{thm:5.7} implies in particular that every
  bounded unitary representation of the Hilbert--Lie group
  $G = \U_2(\sH)$ is a direct sum of irreducible ones, and these
  are the well-known representations arising from the decomposition
  of the tensor products ${\sH^{\otimes n} \otimes (\sH^*)^{\otimes m}}$
  (\cite{Ki73} and \cite{Ne98, Ne14b}).

  However, as shown by Boyer in \cite{Bo80},
  the group $\U_2(\sH)$ has unbounded unitary
  factor representations not of type I. In this respect $\U_2(\sH)$
  behaves very different from a compact group.
  See also Subsection~\ref{rem:factorrep} for a concrete description of such representations.
\end{rem}

\begin{rem} \mlabel{rem:toroid}
  (Toroidal groups) Let $(w_n)_{n \in \N}$ be a sequence of positive numbers 
and $V := \ell^2(\N,w,\R)$ the corresponding weighted $\ell^2$-space, 
consisting of complex sequences $(x_n)$ for which 
\[ \|x\|_{2,w}^2 = \sum_n w_n |x_n|^2 < \infty.\] 
Then the group $\Gamma := \Z^{(\N)}$ is discrete in $V$ if and only 
if $\inf (w_n)_{n\in \N} > 0$. Suppose that this is the case 
and consider the abelian Lie group $T := V/\Gamma$. 
Then $\lambda = (\lambda_n) \in V'$ is equivalent to 
\[ \|\lambda\|_{2,w^{-1}}^2 = \sum_n \frac{1}{w_n} |\lambda_n|^2 < \infty
\quad \mbox{ because } \quad 
\sup \Big\{ \sum_n \lambda_n x_n \: \|x\|_{2,w} \leq 1 \Big\} 
= \|\lambda\|_{2,w^{-1}}.\]  
Such a sequence $(\lambda_n)$ defines a character 
$\chi_\lambda(x) = e^{2\pi i \sum_k \lambda_k x_k}$ of $T$ if 
and only if $\lambda_n \in \Z$ for every $n \in \N$, i.e.,
\[ \hat T \cong \Z^\N \cap \ell^2(\N, w^{-1},\R).\] 

Each $\lambda \in \hat T = \Hom(T,\T)$ is finitely supported if and only if 
the weight sequence $(w_n)$ is bounded from above, which is equivalent to the
boundedness of of the sequence of norms $\|e_n\|_{2,w}$ of the 
generators of~$\Gamma$.
In this case it takes values in some interval $[a,b] \subeq \R_+$,
so that $V = \ell^2$ and $\hat T \cong \Z^{(\N)}$ is  countable.
Accordingly, every norm continuous representation of
$T$ (which is given by a spectral measure on $\hat T$) is a direct
sum of irreducible (hence one-dimensional) representations
(cf.~Theorem~\ref{thm:z-discrete}).
So we have an analogy between the representation theory of
these groups for a bounded sequence of generators of $\Gamma$
and the representation theory of Hilbert--Lie algebras
for which $\Delta^\vee$ is bounded
(Theorems~\ref{thm:5.7} and \ref{thm:5.7b}).
\end{rem}

\section{Automorphism groups and derivations}
\mlabel{sec:4}

Starting with this section, we shall restrict our attention
to the simple Hilbert--Lie algebra 
$\g = \fu_2(\sH)$, where~$\sH$ is a separable infinite-dimensional
complex Hilbert space. We want to understand
which restrictions this imposes
for a unitary representation $(\pi, \cH)$ of $G = \U_2(\sH)$
to be extendable to a representation of 
$G \rtimes_\alpha \R$ for a one-parameter group
$(\alpha_t)_{t \in \R}$ of automorphisms of $G$.
We know that,  for any such $\alpha$, there exists a selfadjoint operator
$H$ on $\sH$ with
\[ \alpha_t(g) = e^{itH} g e^{-itH} \quad \mbox{ for }\quad
  t\in \R, g \in G\]
(Theorem~\ref{thm:aut-grp}). 
Therefore the problem is non-trivial if and only if
$H$ is not a Hilbert--Schmidt operator.

In Subsection~\ref{subsec:WvN}  we prove a rather general 
Perturbation Theorem (Theorem~\ref{thm:4.4})
that can in particular be used to reduce questions concerning
covariance of unitary representations with respect to some $\alpha$
to  the case where $H$ is diagonalizable
(Corollary~\ref{cor:5.3}). This observation is based on 
a Weyl--von Neumann decomposition $H = H_d + H_{\rm HS}$, where
$H_d$ is diagonalizable and $H_{\rm HS}$ is Hilbert--Schmidt.
We may  assume henceforth that $H$ is diagonalizable,
which implies that it fixes a maximal abelian subalgebra
$\ft \subeq \g$ pointwise. We write
  \[ T := \exp \ft \subeq G \]
  for the corresponding ``maximal torus'', an abelian connected
  Banach--Lie group. 
Then we show that any factorial unitary representation
$(\pi,\cH)$ that is a {\it $T$-weight representation},
i.e., has a basis of $T$-eigenvectors,
extends to $G \rtimes_\alpha \R$ (Proposition~\ref{prop:4.11}).
As these results apply equally well to projective representations,
we briefly discuss in Subsection~\ref{subsec:centext} some techniques
for the identification of the corresponding cocycles.

For representations of $G \rtimes_\alpha \R$,  we can ask for the existence
of ground states, but
also investigate the semiboundedness of the extended representation,
or conditions
for $U_t := \hat\pi(e,t)$ to have spectrum which is bounded from below.
To this end, we need to understand the structure of the subgroup 
$G^0 = \Fix(\alpha)$ of $\alpha$-fixed points in $G$,
which is described in Subsection~\ref{subsec:galpha}.
Moreover, we need to understand which representations
of $G^0$ can possibly arise on the minimal $U$-eigenspace. 
As $G$ contains a dense subgroup $G_{\rm alg}$ which is a direct
limit of compact Lie groups, the techniques developed in
\cite{NR24} provide significant necessary conditions
(Theorem~\ref{thm:4.20} in Subsection~\ref{subsec:4.4}).

In Subsection~\ref{rem:factorrep}, we discuss a class of 
representations that is a very special case of a 
construction that has been used in \cite{Bo80} to show that
$\U_2(\sH)$ has factor representations of type II$_1$. 
For our context, this provides representations
of $\U_2(\sH)$ with no $T$-eigenvector which are covariant
for the full automorphism group. 


\subsection{The Weyl--von Neumann decomposition}
\mlabel{subsec:WvN}

The following result can be found in \cite[App.~A]{DA68}.
It is originally due to John von Neumann and Hermann Weyl.
More recent generalizations
can be found in  \cite{Be71} and \cite{Vo79}.

\begin{thm} \mlabel{thm:vNW-dec} {\rm(Weyl--Von Neumann Decomposition
    Theorem, \cite{JvN35, We09})}
Let $A$ be a selfadjoint operator on a separable 
complex Hilbert space. Then there exists a diagonalizable 
selfadjoint operator $D$ and a Hilbert--Schmidt operator 
$T$ with $\|T\|_2 < 1$, such that 
\[ A = D + T.\] 
\end{thm}

\begin{thm} {\rm(Trotter Product Formula; \cite[Ex.~III.5.11]{EN00})}
  \mlabel{thm:trotter}
  Let $(T_t)_{t \geq 0}$ be a strongly continuous semigroup with
  generator $A$ on the Banach space $E$ and
  $B \in B(E)$ a bounded operator. Then
  the semigroup $(U_t)_{t \geq 0}$ generated by $A + B$ is given by
  the Trotter--Product Formula
  \[ U_t x = \lim_{n \to \infty} \Big( T_{t/n} e^{t B/n}\Big)^n x \quad\mbox{ for } \quad     x \in E.\]
\end{thm}

\begin{thm} \mlabel{thm:4.4}
  {\rm(Perturbation of one-parameter groups)}
  Let $G$ be a connected Banach--Lie group and
  $(\alpha_t)_{t \in \R}$ be a one-parameter group in $\Aut(G)$,
  defining a continuous action of $\R$ on $G$, so that $G \rtimes_\alpha \R$
  is a topological group. Let $(\beta_t)_{t \in \R}$
be a second such one-parameter group   of automorphisms
  of $G$ such that the infinitesimal generators
  $D_\alpha$ and $D_\beta$ on $\g$ satisfy
  \[ D_\beta = D_\alpha + \ad x \quad \mbox{  for some  } \quad x \in \g.\] 
  Then $G \rtimes_\alpha \R \cong G \rtimes_\beta \R$ 
  as topological groups, and the isomorphism can be chosen to be
  $C^1$-for the product manifold structure. 
\end{thm}

\begin{prf} We consider the
  half Lie group $\hat G := G \rtimes_\alpha \R$,
  i.e., $\hat G$ is a smooth manifold
  and a topological group, for which left multiplication maps
  are smooth   (see \cite{MN18} and \cite{BHM23} for more
  on half Lie groups).
  Let $\hat\g := \g \oplus \R$ be its tangent space
  (which is not a Lie algebra if $D_\alpha$ is unbounded, i.e.,
  $\alpha$ is not norm-continuous). 
  According to \cite[Cor.~4.9]{MN18},
  the element $\bd := (x, 1) \in \hat\g$ generates a
  $C^1$-one-parameter group $\delta\: \R \to \hat G$ with
  $\delta'(0) = \bd$. Now \cite[Cor.~4.15]{MN18} implies that
  the one-parameter groups 
  $\gamma_1(t) := (\exp(tx),0)$ and $\gamma_2(t) = (0,t)$
  satisfy 
  \[ \delta(t) = \lim_{n \to \infty} \big(\gamma_1(t/n)\gamma_2(t/n)\big)^n.\]
  For the continuous action on $\g$ we thus obtain
  \[ \Ad(\delta(t))
    = \lim_{n \to \infty} \big(\Ad(\gamma_1(t/n))\Ad(\gamma_2(t/n))\big)^n
    = \lim_{n \to \infty} \big(e^{\textstyle{\frac{t}{n} \ad x}}\L(\alpha_{t/n})\big)^n.\]
Theorem~\ref{thm:trotter} now implies that
$\Ad(\delta(t)) = \L(\beta_t)$ and hence that
\[ \beta_t(g) = \delta(t) g \delta(t)^{-1} \quad \mbox{ for } \quad
  t\in \R, g \in G \]
because $G$ is connected. We conclude that the map 
  \[ \Phi  \: G \times_\beta \R \to G \rtimes_\alpha \R, \quad
    (g,t) \mapsto (g,0) \delta(t) \] 
  is an isomorphism of topological groups.

  Writing $\delta(t) = (\delta_G(t), t)$ with a $C^1$-curve
  $\delta_G \:\R \to G$, we have
  $\Phi(g,t) = (g \delta_G(t), t)$
  and $\Phi^{-1}(g,t) = (g \delta_G(t)^{-1}, t)$, showing
  that $\Phi$ actually is a $C^1$-diffeomorphism.
\end{prf}

We record an important obvious consequence of the preceding theorem. 
\begin{cor} \mlabel{cor:4.5}
  In the context of {\rm Theorem~\ref{thm:4.4}}, a
  continuous homomorphism $\pi \: G \to U$
    to some topological group $U$ extends to a continuous
    homomorphism $G \rtimes_\alpha \R \to U$ if and only if
  this is the case for  $G \rtimes_\beta \R$.
\end{cor}

In the context of unitary and projective unitary representations,
Corollary~\ref{cor:4.5} can be formulated as follows. We refer to \cite{Ne14c, JN19} for a
discussion of continuity and smoothness of projective representations. 
\begin{cor} \mlabel{cor:4.6}
 A (projective) unitary representation $(\pi, \cH)$ of $G$
  extends to a continuous (projective) representation
 of $G \rtimes_\alpha \R$ if and only if
  this is the case for $G \rtimes_\beta \R$.
\end{cor}

\begin{prf} We only have to apply the preceding corollary to
  the topological groups $U = \U(\cH)_s$ (the unitary group, endowed with
  the strong operator topology) and to its quotient
  $\PU(\cH)_s = \U(\cH)_s/\T \1$.
\end{prf}

Eventually, we apply these general results to the Hilbert--Lie group
$G = \U_2(\sH)$. 
\begin{cor} \mlabel{cor:5.3}
{\rm(Perturbation of implementable one-parameter groups)}
  Let $(\pi, \cH)$ be a (projective) unitary 
  representation of $G = \U_2(\sH)$ and 
$\alpha_t(g) = e^{itH} g e^{-itH}$ for $H = H^*$ on $\sH$. 
Write $H = H_1 + H_2$, where $H_1$ and $H_2$ are selfadjoint
with $\| H_2\|_2 < \infty$. Then
  $\pi$ extends to $G \rtimes_\alpha \R$ if and only if
  this is the case for $G \rtimes_\beta \R$, where
$\beta_t(g) = e^{itH_1} g e^{-itH_1}$.   
\end{cor}

Note that this theorem applies in particular to the case where
$H = H_d + H_{\rm HS}$ is a Weyl--von Neumann decomposition, where
$H_d$ is diagonalizable and $H_{\rm HS}$ is Hilbert--Schmidt.
It reduces the problem to implement
one-parameter groups of automorphisms in (projective) 
unitary representation of $\U_2(\sH)$ to the case where
these one-parameter groups are diagonalizable on $\sH$.

\begin{rem} The Weyl--von Neumann decomposition
  is far from unique in any sense.
  For example we can move a summand $A$ of $H_d$ that corresponds to a
  finite sum of finite-dimensional eigenspaces of $H_d$ to $H_{\rm HS}$
  and obtain the new decomposition
  \[ H = (H_d - A) + (A + H_{\rm HS}).\]
  If $H_d$ has only finitely many finite-dimensional eigenspaces, we
  can thus get rid of all finite-dimensional eigenspaces of~$H_d$.

  Note that, for $J = \N$ and $H = \diag(n^{-1/2})$, all eigenspaces
  are one-dimensional, but $H$ is not Hilbert--Schmidt. This operator
  is not equivalent to one for which all eigenspaces are
  infinite-dimensional. 
  \end{rem}

\begin{prob} Extend the Weyl--Von Neumann Decomposition
  to skew-hermitian operators on real and quaternionic Hilbert spaces.
  Over $\C$, this is trivially obtained by multiplication with $i$.
  For a skew-adjoint operator $D$ on the real Hilbert space $\sH$,
  we can complexify and write $D = D_d + D_{HS}$ on $\sH_\C$,
  but then it is not clear if $D_d$ (or $D_{HS}$) can be chosen
  in such a way as to commute with the complex conjugation $\tau$
  of $\sH_\C$ with respect to~$\sH$. A similar problem arises
  in the quaternionic case. 
\end{prob}

  \begin{rem} \mlabel{rem:4.9} If two one-parameter groups
  $\alpha,\beta \: \R \to \Aut(\U_2(\sH))$ are generated by
  diagonal operators $H_\alpha = \diag(\lambda_j)$ and $H_\beta = \diag(\mu_j)$,
  we call them {\it equivalent},
  denoted $\alpha \sim \beta$, if there exists a $c \in \R$ such that
  $H_\alpha - H_\beta - c \1$ is a Hilbert--Schmidt diagonal operator, i.e.,
  \[ \sum_{j \in J} (\lambda_j - \mu_j - c)^2 < \infty.\]
  Then the equivalence classes are parametrized by the linear quotient space
  \[ \R^J/(\ell^2(J,\R) + \R \1).\]
  In view of the Weyl--von Neumann decomposition and
  Corollary~\ref{cor:4.6}, this is the natural parameter space
  for one-parameter groups of automorphisms when
  studying extensions of (projective) unitary representations.
\end{rem}

\begin{prop} \mlabel{prop:4.10}
  Let $(U_t = e^{itH})_{t\in \R}$ be a unitary one-parameter
  subgroup of $\U(\sH)$   and $\alpha_t = \Ad(U_t)$ be the corresponding
  one-parameter group of automorphisms of $\U_2(\sH)$.
  Then the following are equivalent:
  \begin{itemize}
  \item[\rm(a)] $H$ is diagonalizable.
  \item[\rm(b)] $\alpha$ leaves a maximal abelian
    subalgebra $\ft \subeq \fu_2(\sH)$
    invariant. 
  \item[\rm(c)] $\alpha$ fixes a maximal abelian
    subalgebra $\ft \subeq \fu_2(\sH)$     pointwise. 
  \end{itemize}
\end{prop}

\begin{prf} (a) $\Rarrow$ (b) follows if we choose $\ft$ in such a
  way that it preserves the eigenspaces of $H$.

  \nin (b) $\Rarrow$ (c): The $\ft$-weight spaces in $\sH$ are
  $1$-dimensional (Example~\ref{ex:d.1a}).
  Hence they are permuted by the unitary
  one-parameter group $U_t$, and the same argument as
  in the proof of Lemma~\ref{lem:pres-ideals} implies that
  they are all invariant under $U_t$. Hence $U_t$ commutes with $\ft$.

  \nin (c) $\Rarrow$ (a) follows from the fact that $U_t$ leaves all
the one-dimensional $\ft$-weight spaces invariant.
\end{prf}

For a general Hilbert--Lie algebra $\g = \fz(\g) \oplus \hat\oplus_{j \in J} \g_j$
as in Schue's Decomposition Theorem~\ref{thm:1.3} 
and a continuous one-parameter group $(\alpha_t)_{t \in \R}$
of automorphisms, we know from Lemma~\ref{lem:pres-ideals}
that it preserves all simple ideals and the center.
Its restriction to a finite-dimensional simple ideal $\g_j$ is given by
inner automorphisms, hence preserves a maximal abelian subalgebra
$\ft_j$. The preceding proposition characterizes this requirement
for the infinite-dimensional
simple ideals of the type $\fu_2(\sH)$, where $\sH$ is an
infinite-dimensional complex Hilbert space.
We plan to investigate the context where all simple ideals
are finite-dimensional in \cite{NR24b}.

Fix an orthonormal basis $(e_j)_{j \in J}$ of $\sH$
and a corresponding maximal abelian subalgebra $\ft \subeq \g = \fu_2(\sH)$.
In the following we  write
    \[ G_{\rm alg} := \U_2(\sH)_{\rm alg} \]
    for the subgroup consisting of unitary operators $g$ for which
    the matrix of $g - \1$ has only finitely many non-zero entries.

\begin{prop} \mlabel{prop:4.11}
{\rm(Extension of weight representations)}
Let $(\pi, \cH)$ be a factorial representation of
  $G = \U_2(\sH)_{\rm alg}$ or of $\U_2(\sH)$.
  Suppose that the restriction of $\pi$ to the subgroup 
  $T \subeq G$ of diagonal matrices is diagonalizable.
  We assume that $\alpha_t(g) = e^{t \bd} g e^{-t\bd}$ for
  $\bd = \diag(i d_j)_{j \in J}$. Then the following assertions hold:
  \begin{itemize}
  \item[\rm(a)] There exists an
  extension $(\hat\pi,\cH)$ of $\pi$ to the semidirect product
  $\hat G = G \rtimes_\alpha \R$ such that 
  the operators  $U_t := \hat\pi(e,t), t \in \R$, are
  contained in $\pi(T)''$, i.e., they are scalar on all $T$-weight spaces.  
\item[\rm(b)] If the spectrum of $(U_t)_{t \in \R}$ is bounded from  below
  and $\cP_\pi \subeq \hat T$ is the weight set of $\pi\res_T$,   then 
  \begin{equation}
    \label{eq:wsembo}
    \inf  \la \cW\lambda - \lambda, -i\bd \ra  > - \infty
    \quad \mbox{ for } \quad \lambda \in \cP_\pi, 
  \end{equation}
  where $\cW$ is the Weyl group of $\Delta(\g,\ft)$. 
  \end{itemize}
\end{prop}

\begin{prf} As $\U_2(\sH)_{\rm alg}$ is a dense subgroup of  $\U_2(\sH)$,
  for every representation $\pi$ of $\U_2(\sH)$, the restriction
  to $\U_2(\sH)_{\rm alg}$ generates the same von Neumann algebra.
  Likewise the weight set $\cP_\pi$ of $T \subeq \U_2(\sH)$ can by
  restriction to  $T_{\rm alg}$ be identified with a subset of
  $\hat{T_{\rm alg}}$. It therefore suffices for the proof to assume that
  $G = \U_2(\sH)_{\rm alg}$ and that $T \subeq G$ is the corresponding subgroup
  of diagonal matrices. 

\nin  (a) We have $\L(\alpha_t) x_\alpha = e^{it \chi_\bd\alpha)} x_\alpha$ for
    $x_\alpha \in \g_\C^\alpha,$ 
  where
  \[ \chi_\bd\: \Z[\Delta] := \Spann_\Z \Delta \to (\R,+) \]
  is a group homomorphism (\cite[Lemma~6.2]{St01}).
In terms of $\bd$, it is  given
  by $\chi_\bd(\eps_j - \eps_k) = d_j - d_k,$ resp., 
\[     \chi_\bd \: \R^{(J)} \supeq \hat T \to \R, \quad
    \chi_\bd(\bx)  =  \sum_j d_j x_j.\]  

  Let $\cP_\pi \subeq \hat T$ be the set of $T$-weights of $\pi$
  and fix $\nu_0 \in \cP_\pi$. We claim that factoriality of $\pi$
  implies that
  \[ \cP_\pi \subeq \nu_0 + \Z [\Delta].\]
  To see this, we consider the subspace
  \[ \cH_0 := \hat\bigoplus_{\alpha - \nu_0 \in\Z [\Delta]} \cH_\alpha(T),\]
  which is easily seen to be invariant under all compact subgroups
  of $G = \U_2(\sH)_{\rm alg}$ corresponding to finite subsets of~$J$, hence
  under   $G$ itself.   It is certainly invariant under the commutant $\pi(G)'$.
  Factoriality means that
  \[ \C \1 = \cZ(\pi(G)'') = \pi(G)'' \cap  \pi(G)'
    = (\pi(G)' \cup \pi(G))'.\]
  So the   projection onto $\cH_0$ is contained in the
  trivial center of the von Neumann algebra $\pi(G)''$.
  Since it is non-zero, it must be $\1$, and thus
    $\cH = \cH_0$.
    We may therefore define a unitary one-parameter group on $\cH$ by
  \begin{equation}
    \label{eq:extform}
 U_t v_\nu := e^{it \chi_\bd(\nu - \nu_0)} v_\nu \quad \mbox{ for } \quad
 v_\nu \in\cH_\nu(T).
\end{equation}

\nin (b) follows immediately from~\eqref{eq:extform}. 
\end{prf}

\begin{rem}  In \cite[Ex.~V.9]{Ne04} one finds examples of weight modules
    without extremal weights which define holomorphic 
    representations of $\tilde\GL_1(\sH)$. 
    On the algebraic level, a classification of
    weight modules of $\g_{\rm alg}$ has been obtained by
    Dimitrov and Penkov in~\cite{DP99}.
  \end{rem}

    \begin{prob} Determine which unitary weight representations of
      $\U_1(\sH)$ extend continuously to $\U_2(\sH)$.    
\end{prob}


\subsection{Projective representations and central extensions}
\mlabel{subsec:centext}

Many results of the preceding subsection apply equally well
to projective representations.
This motivates the brief discussion, in the present subsection,
of some techniques
for the identification of the corresponding cocycles. Here an important
feature of Hilbert--Lie algebras is that the invariant scalar product
provides a bijection between (bounded) derivations and $2$-cocycles.
For projective representations with a $T$-fixed ray, this permits us
to determine the corresponding Lie algebra cocycle rather directly.

If $\oline\pi \:  G\to \PU(\cH)$ is a {\it smooth projective
representation} of a Lie group~$G$, i.e., the set of all rays in
the projective space $\bP(\cH)$ with smooth orbit map is dense, then we obtain a
central Lie group extension
\begin{equation}
  \label{eq:gsharpx}
  G^\sharp := \{ (U,g) \in \U(\cH) \times G \: \oline U = \oline\pi(g)\}
\end{equation}
of $G$ by the circle group $\T$, see \cite{JN19}. Then $\pi(U,g) := U$ defines a
unitary representation of $G^\sharp$, lifting the
projective representation $\oline\pi$ of $G$.

Corollary~\ref{cor:5.3} has an interesting analog on the level of central
extensions. Any central Lie group extension $G^\sharp$ of
$G = \U_2(\sH)$, leads to a central
Lie algebra extension
\[ \g^\sharp = \R \oplus_\omega \g\quad \mbox{ with } \quad
  [(t,x),(s,y)] = (\omega(x,y),[x,y]),\]
where $\omega \: \g \times \g \to \R$ is a continuous Lie algebra cocycle.
In view of \cite[Lemma~VII.1]{Ne04}, for $\g = \fu_2(\sH)$, the
continuous Lie algebra cocycles can all be written in the form
\[ \omega_\bd(x,y) = \tr(\bd[x,y]), \quad \bd\in \fu(\sH), x,y \in \fu_2(\sH).\]
As $\omega_\bd = 0$ is equivalent to $\bd \in\R i \1$,
the cocycle $\omega_\bd$ is a coboundary if and only if
$\bd \in \fu_2(\sH) + \R i \1$. This leads to the isomorphism
\[ H^2(\fu_2(\sH),\R) \cong \fu(\cH)/(\R i \1 + \fu_2(\sH))\]
(cf.\ also \cite[Ex.~2.18]{Ne02c} and Remark~\ref{rem:4.9}). 
Theorem~\ref{thm:vNW-dec} now implies that each cohomology class
is represented by a cocycle 
$\omega_\bd$ for which $\bd$ is diagonalizable. Since
the second homotopy group $\pi_2(\UU_2(\sH))$ 
vanishes (combine \cite[Thms.~II.12, II.14]{Ne02} on
the homotopy groups of $\GL_2(\sH)$ with the
polar decomposition \cite[Prop.~A.4]{Ne04}), 
all Lie algebra cocycles integrate to central Lie group
extensions of the associated simply connected Lie group
(\cite{Ne02b}).

A {\bf concrete model} for the central extension of $\U_2(\sH)$,
corresponding to a diagonal operator $\bd$, can be
described as follows. 
The group $G  =\U_2(\sH)$
contains the abelian Banach--Lie group $T$ of diagonal operators.
We shall also consider a Banach--Lie group $\U_{1,2}(\sH)$, 
whose Lie algebra can be specified as follows.
We write $x \in \fu_2(\sH)$ as $x = x_d + x_r$, where $x_d \in \ft$
is a diagonal operator and $x_r \in \ft^\bot$ is off-diagonal.
One easily shows that, for $x,y \in \fu_2(\sH)$, the operator
$[x,y]$ is trace class, and  the same holds for its diagonal part.
Therefore
\[ \fu_{1,2}(\sH) :=
  \Big\{ x \in \fu_2(\sH) \: \sum_{j \in J} |x_{jj}| < \infty \Big\} \]
is a Banach--Lie algebra with respect to the norm
$\|x\| := \|x \|_2 + \|x_d\|_1$. We write
$\U_{1,2}(\sH) \subeq \U_2(\sH)$ for the corresponding Banach--Lie group
with Lie algebra $\fu_{1,2}(\sH)$. We now obtain
a surjective morphism
\[ \U_{1,2}(\sH) \rtimes T \onto \U_2(\sH), \quad
(g_1, g_2) \mapsto g_1 g_2 \]
of connected Banach--Lie groups, 
whose kernel is $\{(g,g^{-1}) \: g \in T_1 \}$, where
\[ T_1 := T \cap \U_1(\sH) \]
is the group of diagonal operators in $\U_1(\sH)$. 
For any character $\chi  \: T_1 \to \T$, we now obtain a central extension
\[ q \: \U_2(\sH)^\sharp_\chi := 
  (\U_{1,2}(\sH) \rtimes T)/\ker(\chi) \to \U_2(\sH),
  \quad [(g,t)] \mapsto gt,\]
whose kernel is mapped to the circle by 
$T_1/\ker(\chi) \to \T, [(g,g^{-1})] \mapsto \chi(g).$
We pick $\bd \in \fu(\sH)$ such that 
\[ \chi(\exp x) =  e^{2\pi i \tr(\bd x)} \quad \mbox{ for } \quad x \in \ft_1.\]

On the Lie algebra level, we have the central extension 
\[ \g^\sharp = (\fu_{1,2}(\sH) \rtimes \ft)/\ker(\dd\chi) \to \R,\]
for which
\[ \sigma \: \g \to \g^\sharp, \quad x = x_d + x_r  \mapsto   (x_r, x_d) \]
is a continuous linear section. A corresponding cocycle is 
\begin{align*}
  \omega(x,y)
  &= [\sigma(x),\sigma(y)] - \sigma([x,y])
  = [(x_d,x_r), (y_d,y_r)] - ([x,y]_r, [x,y]_d) \\
  &  = ([x_r, y_d] + [x_d, y_r] + [x_r,y_r],0) - ([x_d, y_r] + [x_r, y_d]
    + [x_r, y_r]_r, [x_r,,y_r]_d) \\
  &  = ([x_r,y_r]_d, - [x_r,y_r]_d),
\end{align*}
which corresponds to $\dd\chi([x_r, y_r]_d) \in \R \cong \L(\T)$.
With $\exp_\T(x) = e^{2\pi i x}$, this translates for
$\chi(\exp x) = e^{2\pi i\tr(\bd x)}$ into 
\begin{equation}
  \label{eq:omega-formel}
 \omega(x,y)
 = \tr(\bd[x_r, y_r])
 = \tr(\bd[x, y]) = \omega_\bd(x,y).
\end{equation}

\begin{prob} It is an interesting question
  which central extensions of  $\U_2(\sH)$ arise
  from projective unitary representations.
  In Subsection~\ref{subsec:5.4} below    we shall see examples, where   the Lie algebra cocycles satisfy certain integrality conditions,
  but it is not clear if this covers all cocycles coming from
  projective unitary representations.
\end{prob}

\begin{ex} \mlabel{ex:4.13}
  Let $(\oline\pi,\cH)$ be a smooth projective unitary
  representation of $G= \U_2(\sH)$ and
  $\pi \: G^\sharp \to G$ be the corresponding unitary representation
  of the central extension $G^\sharp$ (cf.\ \eqref{eq:gsharpx}).
  We assume that $\Omega \in \cH$ is a smooth unit vector
 that is an eigenvector of $T^\sharp$.

  As above, we write $T^\sharp$ as $(T_1 \times T)/\ker(\chi_\bd)$,
  so that the inclusion $\T \into T^\sharp$ corresponds to
  the isomorphism $T_1/\ker(\chi_\bd) \to \T, [t] \mapsto \chi_\bd(t)$.
  Further, the map $T^\sharp \to \T, [(t_1, t)] \mapsto \chi_\bd(t_1)$,
  has kernel $[\{e\} \times T] \cong T$, so that
  \[ T^\sharp \cong \T \times T.\]

  We now consider the continuous linear section
  $\sigma \: \g \to \g^\sharp$, specified by
  $\dd\pi(\sigma(x))\Omega  \bot\Omega.$
  For the corresponding cocycle 
    \[   \omega_\Omega(x,y) \1= [\sigma(x),\sigma(y)] - \sigma([x,y]),\]
    we then find
    \begin{align*}
 i \omega_\Omega(x,y)
& = i\la \Omega, \omega_\Omega(x,y) \Omega \ra 
      = \la \Omega, \dd\pi([\sigma(x), \sigma(y)]) \Omega \ra \\
&      = -\la \dd\pi(\sigma(x)) \Omega, \dd\pi(\sigma(y))\Omega \ra 
       + \la \dd\pi(\sigma(y)) \Omega, \dd\pi(\sigma(x))\Omega \ra \\
&   = 2i \Im \la \dd\pi(\sigma(y)) \Omega, \dd\pi(\sigma(x))\Omega \ra.      
    \end{align*}
    By construction, the diagonal algebra $\ft$, for which $\Omega$
    is an eigenvector, is contained in the radical of this form.
    Let $x = x_d + x_r$ denote the decomposition of $x \in \g = \fu_2(\sH)$
    into diagonal component $x_d \in \ft$ and off-diagonal component~$x_r$.
    Then we find
    \[ \omega_\Omega(x,y) = \omega_\Omega(x_r, y_r).\]

Now let $\rho \: \fu_{1,2}(\sH) \to \End(\cH^\infty)$ be the
    representation of $\fu_{1,2}(\sH)$ on $\cH^\infty$,
    obtained from splitting the central extension.
    Then there exists a continuous linear functional
    $\lambda \in i\ft_1' \cong \ell^\infty(J,\R)$ such that
    \[ \rho(x) \Omega = i\lambda(x) \Omega \quad \mbox{ for } \quad
      x \in \ft_1.\] 
    Then $x_r, y_r \in \fu_{1,2}(\sH)$ and  $\Omega \bot \rho(x_r)\Omega$ imply that 
    $\dd\pi(\sigma(x_r)) = \rho(x_r)$.  We thus obtain
\begin{align*}
  \omega_\Omega(x,y)
&  = \omega_\Omega(x_r, y_r) 
   = 2i \Im \la \rho(y_r)\Omega, \rho(x_r)\Omega \ra
   =  \la \Omega, [\rho(x_r), \rho(y_r)]\Omega \ra\\
&   =  \la \Omega, \rho([x_r, y_r]) \Omega \ra
   =  \la \Omega, \rho([x_r, y_r]_d) \Omega \ra
   =  i\lambda([x_r, y_r]_d).
\end{align*}
Extending $\lambda \in i\ft_1' \cong \ell^\infty(J,\R)$ 
to a continuous linear functional on $\g_1 = \fu_1(\sH)$,
this takes the form
\begin{align*}
  \omega_\Omega(x,y)
& =  i\lambda([x_r, y_r]) =  i\lambda([x, y]).
\end{align*}
For $\lambda(x) = -i\tr(\bd x)$, $\bd \in \fu(\sH)$, we
derive from \eqref{eq:omega-formel} that
$\omega_\Omega = \omega_\bd$. 
\end{ex}

\subsection{The fixed point group} 
\mlabel{subsec:galpha}

Let $G = \U_2(\sH)$, $(e_j)_{j \in J}$ be a Hilbert-Lie group on $\sH$,
and $\bd =\diag (i d_j)_{j \in J}$ a purely imaginary diagonal 
operator, generating the one-parameter group
$\alpha_t(g) = e^{t \bd} g e^{-t\bd}$ of automorphisms of~$G$.
In this section we briefly discuss the structure of the subgroup
\[ G^0 := G^0(\alpha) := \Fix(\alpha) = \{  g \in G \:
  (\forall t \in \R)\, \alpha_t(g) = g \} \] 
of $\alpha$-fixed points.

An operator is fixed by  $\alpha$ if and only if it commutes with
$\bd$, i.e., preserves all $H$-eigenspaces
\[ \sH_\lambda := \ker(\bd - \lambda\1).\] 
We thus obtain
\[ \g^0 := \g^0(\bd) 
  = \hat\bigoplus_{\lambda \in i\R} \fu_2(\sH_\lambda).\]
Here infinite-dimensional eigenspaces correspond to infinite-dimensional
simple ideals of $\g^0(\bd)$, and each finite-dimensional eigenspace
contributes a one-dimensional subspace to the center:
\[ \fz(\g^0) 
  = \hat\bigoplus_{\dim \sH_\lambda < \infty} \R i \1_{\sH_\lambda}.\]
Let $z_\lambda := 2\pi i \1_{\sH_\lambda}$. Then
$\exp(z_\lambda) = \1$, so that
\[ Z(G^0)_e = \exp(\fz(\g^0))
  \cong \fz(\g^0)/\Gamma\quad\mbox{ for } \quad 
  \Gamma = \sum_{\dim \sH_\lambda < \infty} \Z z_\lambda.\]
Note that
\[ \|z_\lambda\|_2 = 2\pi \sqrt{\dim \sH_\lambda} \geq 2 \pi.\]
For $z = \sum_\lambda x_\lambda z_\lambda$ with $x_\alpha \in \R$, we then have
\[ \|z\|^2_2 = \sum_\lambda |x_\lambda|^2 \cdot   \|z_\lambda\|_2^2, \]
so that we may consider $w_\lambda := \|z_\lambda\|_2^2 = 4\pi^2 \dim \sH_\lambda$
as a weight (cf.\ Lemma~\ref{lem:infsum}).

The subgroup $\Gamma \subeq \ft$ is discrete,
but Remark~\ref{rem:toroid} applies only if
the dimensions of all finite-dimensional eigenspaces $\sH_\lambda$ are
bounded from above. In this case, it implies with
Theorem~\ref{thm:z-discrete} that 
all bounded representations of $Z(G^0)_e$
decompose discretely. 

\begin{rem} \mlabel{rem:cont-disc-deco} 
  If $\bd$ is not diagonalizable, we write it as
  $\bd = \bd_d \oplus \bd_c$, where $\bd_d$ is diagonalizable, and
  $\bd_c$ has no non-zero
  eigenvectors. Accordingly, we write
  $\sH = \sH_d \oplus \sH_c$. Then $\fz_{\fu_2(\sH)}(\bd_c) = \{0\}$.
Then
\begin{equation}
  \label{eq:contspec}
 \fz_{\fu_2(\sH)}(\bd)
  = \fz_{\fu_2(\sH_d)}(\bd_d) \oplus  \fz_{\fu_2(\sH_c)}(\bd_c) 
  = \fz_{\fu_2(\sH_d)}(\bd_d)
\end{equation}
because all compact hermitian operators on $\sH_c$,
commuting with $\bd_c$, vanish.   
\end{rem}

\subsection{Ground state representations of the direct limit group}
\mlabel{subsec:4.4}

As in the preceding subsection, let
$(e_j)_{j \in J}$ be an orthonormal basis of $\sH$.
We write $\U(\sH)_{\rm alg} \subeq \U(\sH)$ for the subgroup
of all unitary operators for which the matrix
of $g - \1$ has only finitely many non-zero entries. We then
have the inclusions 
\[ \U(\sH)_{\rm alg} \subeq \U_1(\sH) \subeq \U_2(\sH) \subeq \U(\sH).\]
In this paragraph we briefly discuss the case where
$\bd = \diag(i d_j)$ is an arbitrary skew-hermitian diagonal
matrix, considered as the infinitesimal generator
of the automorphism group
\begin{equation}
  \label{eq:dag}
 \alpha \: \R \to \Aut(\U(\sH)_{\rm alg}), \quad
 \alpha_t(g) = e^{t\bd} g e^{-t\bd}.
\end{equation}

\begin{defn}
  A  unitary representation $(\hat\pi, \cH)$ of $\hat G = G \rtimes_\alpha \R$
has the form 
\begin{equation}
  \label{eq:pisharp}
\pi^\flat(g,t) = \pi(g) U_t,
\end{equation}
where $(\pi, \cH)$ is a unitary representation of $G$ and $(U_t)_{t \in \R}$ 
is a unitary one-parameter group on $\cH$ satisfying the covariance condition  
\begin{equation}
  \label{eq:commrel}
 U_t \pi(g) U_{-t} = \pi(\alpha_t(g)) \quad \mbox{ for } \quad 
t \in \R, g \in G.
\end{equation}
Writing $U_t= e^{itH}$ with a selfadjoint operator $H$ 
(Stone's Theorem, \cite[Thm.~13.38]{Ru73}),
we call $\hat\pi$, resp., the pair $(\pi,U)$ 
a {\it positive energy representation of $(G,\alpha)$}  if $H \geq 0$. 
If, in addition, for the {\it minimal energy space} $\cH^0 := \ker H$, 
the subset $\pi(G)\cH^0$ spans a dense subspace of $\cH$, 
we call $(\pi, \cH)$ a {\it ground state representation}.
\end{defn}

We may consider the Hilbert-Lie group  $\U_2(\ell^2)$
in which the subgroup $\U_2(\ell^2)_{\rm alg} = \U_\infty(\C)
= \indlim \U_n(\C)$ carries the structure of a direct limit
Lie group, and 
$\alpha_t \in \Aut(\U_\infty(\C))$ determined by \eqref{eq:dag}.
Then $\U_\infty(\C)^0 \subeq \U_\infty(\C)$
is the subgroup preserving all eigenspaces of
the diagonal operator $\bd$ on $\C^{(\N)}$. 
If the $(d_n)_{n \in \N}$ are mutually different,
  this subgroup coincides with the torus $T_\infty\cong \T^{(\N)}$ 
 of diagonal matrices in $\U_\infty(\C)$. 
Note that $T_\infty$ is a direct limit of the tori 
$T_n := T \cap \U_n(\C) \cong \T^n$. 

The following theorem characterizes those representations 
of the group $G^0_{\rm alg}$
which arise on minimal eigenspaces in positive energy representations.
Its main feature is that it provides a tool to create representations
of $G_{\rm alg}$ with rather specific properties with respect to
$\alpha$ out of representations of the group~$G^0_{\rm alg}$. The
structure of the group $G^0$ has been discussed in Subsection~\ref{subsec:galpha},
but this discussion shows that, if $\bd$ has infinite-dimensional
eigenspaces, then the structure of $G_{\rm alg}^0$ is more complicated
than the structure of $G_{\rm alg}$ itself.

\begin{thm} \mlabel{thm:4.20} A unitary representation $(\pi^0, \cH^0)$ of
$G^0_{\rm alg}$ extends to a ground state representation
 of $G_{\rm alg} \rtimes_\alpha \R$ if and only if 
\[-i \partial \pi^0(i E_{nn} - i E_{mm}) \geq 0 \quad \mbox{ holds for } \quad
    d_n > d_m.\] 
  \end{thm}

  \begin{prf}  This follows from  (36) in \cite[\S 8]{NR24}.
    We only need to observe that $G_{\rm alg}=\U_\infty(\C)$.
\end{prf}

\begin{prob} \mlabel{prob:4.21}
  Show that, for all ground state representations, 
    $\pi\res_T$ is discrete, i.e., a direct sum of $T$-eigenspaces. 
    To study this problem, the case where $\bd$ has only two eigenvalues is
    already interesting. Then $G^0 \cong \U_2(\sH_1) \times \U_2(\sH_2)$,
    where $\sH_{1/2}$ are the $\bd$-eigenspaces in $\sH$, and the
    condition in Theorem~\ref{thm:4.20} imposes rather strong
    spectral conditions for the representation~$\pi^0$.
    If, for instance, $\pi^0$ is trivial on $\U_2(\sH_2)$,
    we find I.~Segal's positivity condition from \cite{Se57}, which
    implies that $\pi^0$ is a direct sum of subrepresentations
    of tensor products $\sH^{\otimes n}$.
    We refer to \cite{Ne14b} for a discussion of these representations.
\end{prob}

\subsection{Some non-weight representations} 
\mlabel{rem:factorrep}

In \cite{Vo76} Voiculescu describes  {\it characters}
(central positive definite functions) of the inductive limit
  group $\U(\sH)_{\rm alg} \cong \indlim \U_n(\C)$
  for $\sH = \ell^2= \ell^2(\N,\C)$, and
  he also shows that these characters extend
to the group $\U_1(\sH)$. 
In this subsection we discuss one case where
the corresponding unitary representation of $\U_1(\sH)$ even extends
to a representation of the Hilbert--Lie group $\U_2(\sH)$
and show that its restriction to the diagonal subgroup $T$
does not decompose discretely. 

All these characters can be described as follows. 
We start with a continuous positive definite function
$p \: \bS^1 \to \C$ satisfying $p(1) = 1$
and construct a class function on $\U_\infty(\C)$ by 
\[  \chi_p(g) := \det \big(p(g)\big) 
  = p(z_1) \cdots p(z_n), \quad \mbox{ where } \quad
  \Spec(g) = \{z_1, \ldots, z_n\},\]
where the spectral values $z_j\not=1$
are enumerated with multiplicities. 
The simplest case is $p(z) = z^k$. Then $\chi_p(g) = (\det g)^k$ 
which extends holomorphically to $\GL_1(\sH)$. 

 For $b > 0$ we consider 
 $q(z) = \frac{1 + bz}{1 + b}$ and
 $p(z) := q(z) q(z^{-1})$. The polynomial $q$ corresponds in some sense to
 representations in antisymmetric tensors (\cite[p.~15]{Vo76})
 and $q(z^{-1})$ to the dual representation.
 For $g = \1 + X \in \U_2(\sH)$, we have
 \begin{align*}
 \frac{(\1 + b g)(\1 + b g^{-1})}{(1 + b)^2} - \1 
&   = \frac{(1 + b^2)\1 + b (g + g^{-1})}{(1 + b)^2} -\1
   = \frac{b (g + g^{-1} - 2\1)}{(1 + b)^2} \\
&   = \frac{b}{(1+b)^2}(X + (\1 + X)^{-1}-\1).
 \end{align*}
 As
 \[ (X + (\1 + X)^{-1}-\1)
   = (\1 + X)^{-1}((\1 + X)(X - \1) + \1)
   = (\1 + X)^{-1}X^2 \in B_1(\sH),\]
we conclude that 
\[  \frac{(\1 + b g)(\1 + b g^{-1})}{(1 + b)^2} \in \1 + B_1(\sH).\] 
Therefore the character 
\[  \chi(g) = \det\big((1 + b)^{-2}(\1 + b g)(\1 + b g^{-1})\big)  \] 
is defined for $g \in \U_2(\sH)$
for a large class of characters 
(see \cite[Thm.~5.2]{Bo80}). 

As $\chi(gsg^{-1}) = \chi(s)$ for $s \in \U_2(\sH)$ and $g \in \U(\sH)$,
the Gelfand--Naimark--Segal (GNS) representation
$(\pi_\chi, \cH_\chi, \Omega_\chi)$ with the
cyclic vector $\Omega_\chi$ extends naturally
to a representation $\hat\pi_\chi$ of the semidirect product
$\U_2(\sH) \rtimes \U(\sH)$, where the representation of $\U(\sH)$
is determined by the covariance relation
\[ \hat\pi_\chi(g) \pi_\chi(h) \hat\pi_\chi(g)^{-1} = \pi_\chi(ghg^{-1})
  \quad \mbox{ and } \quad
  \hat\pi_\chi(g)\Omega_\chi = \Omega_\chi \quad \mbox{ for } \quad
  g \in \U(\sH), h \in \U_2(\sH). \] 

We claim that the restriction to the diagonal subgroup
$T \subeq \U_2(\sH)$ does not decompose discretely.
We have 
\[ \chi(\bt) = \prod_{n = 1}^\infty p(t_n)
 = \prod_{n = 1}^\infty \frac{(1 + b t_n)(1 + b t_n^{-1})}{(1 + b)^2}  
  \quad \mbox{ for } \quad \bt = (t_1, t_2, \ldots), \] 
where the function
\[ p(t) = \frac{(1 + bt)(1+bt^{-1})}{(1+b)^2} = \frac{1 + b^2
  + b(t + t^{-1})}{(1+b)^2}\] 
is the positive definite
function corresponding to the $3$-dimensional representation 
of the circle group $\T$ on $\C^3$, given by
$\pi(t) = \diag(1,t,t^{-1})$ and
\[ p(t) = \la \Omega, \pi(t) \Omega \ra \quad \mbox{ for the unit vector} \quad
  \Omega = \frac{1}{1 + b} \pmat{ \sqrt{1+ b^2} \\ \sqrt{b}\\ \sqrt{b}}.\]
The corresponding measure $\nu$ on $\Z \cong \hat\T$ with
$\hat\nu = p$ is given by
\[ \nu = \frac{1}{(1 + b)^2}((1 + b^2)\delta_0  + b \delta_1+ b \delta_{-1}).\]
Accordingly, the measure $\mu$ on $\Z^\N$ with $\hat\mu = \chi$ on
$T_1 = T \cap \U_1(\sH)$ 
is the infinite tensor product 
$\mu = \nu^{\otimes \N}$ on $\{-1,0,1\}^\N \subeq \Z^\N$, 
and this measure has no atoms. It  is equivalent to the
Haar measure on the Cantor group $(\Z/3\Z)^\N$.
This implies that the spectral measure for the $T$-representation
generated by $\Omega_\chi$ is not atomic. In particular
$\pi_\chi\res_{T}$ does not decompose discretely. 

As we have seen above, $\pi_\chi$ extends for all continuous
$\R$-actions $\alpha$ on $G = \U_2(\sH)$ to the group
$G \rtimes_\alpha \R$. It follows in particular that the extendability
alone does not permit to conclude that $\pi_\chi\res_T$ decomposes
discretely.

\begin{prob} If $\alpha$ is generated by a diagonal operator,
then it fixes $T$, so that the closed subspace generated by
$\pi_\chi(T)\Omega$ is fixed by $(U_t)_{t \in \R}$.
Is it possible to determine the spectrum of these
unitary one-parameter groups on $\cH_\chi$?
As the total subset $\pi_\chi(\U_1(\sH))\Omega$
is invariant under $(U_t)_{t \in \R}$, 
reproducing kernel methods similar to those in \cite[Thm.~II.4.4]{Ne00}
may be useful in this regard.
\end{prob}

  \begin{prob} \mlabel{prob:4.22} Show that, for all non-trivial
  one-parameter groups $(U_t)_{t \in \R}$ of $\U(\sH)$,
  the spectrum of the corresponding representation on $\cH_\chi$
  is neither bounded from below nor from above.   
\end{prob}

\section{Semibounded representations
of Hilbert--Lie groups and central extensions}
\mlabel{sec:5}

In Section~\ref{sec:bounded} we discussed bounded
representations of Hilbert--Lie group. The next larger
class consists of semibounded representations, 
whose theory has been developed 
in  \cite{Ne08, Ne10b, Ne12, Ne14, NZ13}.

We start in Subsection \ref{subsec:reconstruction} with
Theorem~\ref{theoremofreconstruction}, which
asserts that semibounded unitary representations of connected
Hilbert--Lie groups decompose discretely into bounded ones,
hence provide nothing new beyond what we have seen in
Section~\ref{sec:bounded}.
This is complemented by Theorem \ref{thm:b.6}, which  shows
that semibounded projective representations of Hilbert--Lie groups 
can be lifted to representations of the simply connected
covering group. Its proof involves the remarkable Bruhat--Tits Fixed
Point Theorem (see Appendix \ref{app:b}).

As a consequence, enlarging our context to projective semibounded
  representations of $G$ does not lead to  new representations,
but projective representations of the groups
  $G \rtimes_\alpha \R$ for which the spectrum of the one-parameter group
  $\pi(e,t)$ is bounded below are far more interesting. 
As we go along, we learn 
that this is due to the fact that considering $\U_2(\sH)$ as a
Lie group is too restrictive, in particular, when we consider
one-parameter groups $(\alpha_t)_{t \in \R}$ of automorphisms
that act only continuously and not smoothly, i.e., their
infinitesimal generator is unbounded.
This is taken into account in
Subsection \ref{subsec:4.3}, where we show that the existence of a
unitary representation of  $\U_2(\sH) \rtimes_\alpha \R$,
whose restriction to the Lie group
$\U_2(\sH)^\infty \rtimes_\alpha \R$
is semibounded, requires the infinitesimal generator $i\bd$ of $\alpha$
to be semibounded as an operator on $\sH$ 
(Corollary~\ref{cor:5.8}). This result builds on some fine analysis
of the coadjoint action of the semidirect product Lie group
(Theorem \ref{thm:5.8}).
If $\bd$ is unbounded, then
the topological group $\U_2(\sH) \rtimes_\alpha \R$ is not a
Lie group, but it is a smooth Banach manifold and even a
half Lie group in the sense of \cite{MN18},
which has already been used in the proof of
Theorem~\ref{thm:4.4}. Then
$\U_2(\sH)^\infty \rtimes_\alpha \R$ is the Lie group
of all elements for which right multiplications are also smooth,
so that the passage between these two perspectives is well
captured by the half Lie group concept. 

Subsections \ref{subsec:traceclassgroups} and
\ref{subsec:5.4} are devoted to projective representations
of the restricted unitary group $\U_{\rm res}(\sH,D)\subeq \U(\sH)$
 for an operator $D$ on a Hilbert space $\sH$,
where the ``restriction'' is expressed in terms of $gDg^{-1}-D$ being
Hilbert--Schmidt. We recall in Theorem \ref{lem:7.7} (from \cite{Ne04})
that the extremal weight representation
of the trace class group $\U_1(\sH)$, corresponding to a bounded weight,
all extend to projective representations of suitable restricted groups,
which contain in particular the group $\U_2(\sH)$. So we obtain
a large family of projective representations of $\U_2(\sH)$
that define projective semibounded representations of the Lie group 
$\U_2(\sH)^\infty \rtimes_\alpha \R$.
In Theorem~\ref{lem:7.7} we show that the restricted groups are maximal
``covariance groups'' of these representations of $\U_1(\sH)$.
These results provide a new perspective and 
relevant information on significant constructions
which arise naturally in Quantum Field Theory
(see \cite{PS86, Ne10b, KR87}).

\subsection{Smooth and semibounded representations} 
  \mlabel{subsec:reconstruction}
\begin{defn}(\cite{Ne10a})
A unitary 
representation $\pi \: G \to \U(\cH)$ 
is said to be {\it smooth} if the subspace 
$\cH^\infty \subeq \cH$ of smooth vectors is dense. 
This is automatic for continuous 
representations of finite-dimensional groups, but not 
for Banach--Lie groups. For any smooth 
unitary representation, the {\it derived representation} 
\[ \dd\pi \: \g = \L(G)\to \End(\cH^\infty), \quad 
\dd\pi(x)v := \frac{d}{dt}\Big|_{t = 0} \pi(\exp tx)v\]  
carries significant information in the sense that the closure
$\partial \pi(x)$
of the operator $\dd\pi(x)$ coincides with the infinitesimal generator of the 
unitary one-parameter group $\pi(\exp tx)$. We call $(\pi, \cH)$ 
{\it semibounded} if the function 
\begin{equation}
  \label{eq:spi}
 s_\pi \: \g \to \R \cup \{ \infty\}, \quad 
s_\pi(x) 
:= \sup\big(\Spec(i\partial\pi(x))\big)   
\end{equation}
is bounded on the neighborhood of some point in $\g$.  
Then the set $W_\pi$ of all such points 
is an open $\Ad(G)$-invariant convex cone in the Lie algebra $\g$. 
We call $\pi$ {\it bounded} if $s_\pi$ is bounded on some $0$-neighborhood, 
i.e., $W_\pi = \g$.

We shall also use the {\it momentum set $I_\pi$}.
This is the weak-$*$-closed convex hull of the image of the 
momentum map on the projective space of $\cH^\infty$: 
\[  \Phi_\pi \: \bP({\cal H}^\infty)\to \g' \quad \hbox{ with } \quad 
\Phi_\pi([v])(x) 
= \frac{1}{i}  \frac{\la  \dd\pi(x)v, v \ra}{\la v, v \ra}\quad \mbox{ for } 
[v] = \C v. \] 
As a weak-$*$-closed convex subset, $I_\pi$ is completely determined 
by its support functional 
\[  s_\pi \: \g \to \R \cup \{\infty\}, \quad s_\pi(x) 
  = - \inf \la I_\pi,x \ra = \sup(\Spec(i\partial \pi(x))). \]
\end{defn}

We now show that semibounded unitary representations
  decompose discretely into bounded ones, and these have been
  studied in Section~\ref{sec:bounded}.

\begin{thm} \mlabel{theoremofreconstruction} Let $G$ be a connected Hilbert--Lie group.
  Then every semibounded unitary representation
  $(\pi,\cH)$ of $G$ is a direct sum of bounded ones.   
\end{thm}

\begin{prf} 
  Proposition~\ref{prop:2.12} implies that the open
  invariant cone $W_\pi \subeq \g$ intersects the center~$\fz(\g)$,
  so that the restriction $\pi\res_{Z(G)_e}$ is also semibounded.
  We claim that it is a direct sum of bounded representations.
  In fact, by \cite[Thm.~5.1]{Ne08}, the momentum set
  \[ C := I_{\pi\res_{Z(G)}} \subeq \fz(\g)' \] 
  is locally compact with respect to the weak-$*$ topology
  and we obtain a representation of the $C^*$-algebra
  $C_0(C,\C)$ on $\cH$ that commutes with $\pi(G)$.
  Using the corresponding spectral measure on $C$, we may thus decompose
  into representations $(\pi_j)_{j \in J}$, for which $C$ is compact. Then
  $\pi_j\res_{Z(G)}$ is bounded, so that
  $W_{\pi_j} + \fz(\g) = W_{\pi_j}$.
  As $W_{\pi_j}$ intersects the center, it follows that $0 \in W_{\pi_j}$,
  i.e., that $\pi_j$ is bounded.   
\end{prf}

If $\g$ is semisimple, then Proposition~\ref{prop:2.12}
shows as in the proof above
shows that every semibounded representation of $\g$ is bounded.
Then $W_\pi$ intersects $\fz(\g) = \{0\}$, which is equivalent to
boundedness of the representation.

\begin{defn}(\cite{JN19}) We call a smooth projective representation
  $(\oline\pi, \cH)$ of a Lie group $G$ {\it bounded
    (resp.~semibounded)},
  if the corresponding smooth representation 
  $(\pi, \cH)$ of the central Lie group extension
  \[ G^\sharp = \{ (U,g) \in \U(\cH) \times G \: \oline U = \oline\pi(g)\},\]
  given by $\pi(U,g) = U$, is   bounded (resp.~semibounded).   
\end{defn}

The following theorem shows that we cannot expect to find more
semibounded representations of Hilbert--Lie groups
by passing to central extensions.
It builds on quite subtle techniques involving in particular
Lie group cohomology and the Bruhat--Tits Fixed Point Theorem.
For more details, we refer to the proof of Theorem~\ref{thm:b.4}
in the appendix.

\begin{thm}
  \mlabel{thm:b.6}  Let $\g$ be a Hilbert--Lie algebra.
  Then every projective semibounded unitary representation
  $(\oline\pi,\cH)$ of a corresponding simply connected
  Lie group $G$ lifts to a unitary representation of~$G$.
\end{thm}

\begin{prf} We consider the corresponding semibounded
  representation $(\pi,\cH)$ of the central extension 
  $G^\sharp$ of~$G$.
  Then $\partial \pi(1,0) = i \1$ 
  implies that the momentum set $I_\pi
  \subeq (\g^\sharp)'$ 
does not annihilate $(1,0)$. Hence it contains an element 
of the form $\lambda = (z,\alpha)$, $z \not=0$. 
Then $\cO_\lambda$ is semi-equicontinuous (see \eqref{def:seq} in
Appendix~\ref{app:b}), so that the orbit of $\alpha \in \g'$ 
under the affine action 
is semi-equicontinuous (cf.~Proposition~\ref{prop:3.3} and
Remark~\ref{rem:new1}). Now Theorem~\ref{thm:b.4} 
implies that the central Lie algebra extension is trivial.
As $G$ is simply connected, it now follows from\cite{Ne02} that
the central Lie group extension $G^\sharp \to G$ splits. 
\end{prf}


\begin{rem} \mlabel{rem:4.16}
  Let $\hat G = G \rtimes_\alpha \R$
  for a continuous $\R$-action defined by
  $\alpha \: \R \to \Aut(G)$.
  Then $\alpha$ preserves $\fz(\g)$ and all
  simple ideals $(\g_j)_{j \in J}$ (Lemma~\ref{lem:pres-ideals}).

  Suppose that $(\pi,\cH)$ is a positive energy representation of $\hat G$,
  i.e., $U_t := \pi(e,t) = e^{itH}$ with $H = H^*$ bounded from below.
Then
  $U_t \in \pi(G)''$ for $t \in \R$ by the Borchers--Arveson
  Theorem (\cite[Thm.~4.14]{BGN20}).
  If $\pi$ is semibounded, then we may w.l.o.g.\ assume that
  $\bd = (0,1) \in W_\pi$.

  \begin{itemize}
  \item As $\pi(Z(G))$ is central in $\pi(G)''$, we must have
    $\alpha_t(z) = z$ for $t \in \R, z \in Z(G)_e$, provided
    $\ker\pi$ is discrete.

  \item Let $G^0 \subeq G$ be the closed subgroup of $\alpha$-fixed
    points, which is a Hilbert--Lie group (cf.\ Subsection~\ref{subsec:galpha}).
    If $\pi$ is semibounded with $\bd = (0,1) \in W_\pi$,
    then     the restriction of $\pi$ to     
    $G^0 \rtimes_\alpha \R = G^0 \times \R$ 
    is also semibounded.
    The same holds  for the restriction to the abelian subgroup
    $Z \times \R$ for $Z := Z(G^0)_e$, so that \cite[Thm.~4.1]{Ne08} implies
    the existence of a corresponding spectral measure.
    It follows in particular, that the restriction to 
    $Z \times \R$ is a direct sum of bounded representations
    on $G^0$-invariant subspaces.
    We conclude with Proposition~\ref{prop:2.12}, applied to the
      open cone $W_\pi$,   that $(\pi, \cH)$ is a direct sum of bounded representations of $G^0$.

   If, in addition, the dimensions of the finite-dimensional
    $\bd$-eigenspaces are bounded from above, then the discussion
    above 
    shows that $\hat Z$ is countable,
    so that Theorem~\ref{thm:5.7} implies
        that $\pi\res_{G^0_e}$ is a direct sum of irreducible
    bounded representations. 
  \end{itemize}
\end{rem}

\subsection{Semibounded representations require semibounded  $i\bd$} 
\mlabel{subsec:4.3} 

Let $\bd = - \bd^* = i H$ be a, not necessarily bounded,
operator on the complex Hilbert space
$\sH$ and $\g^\infty = \fu_2(\sH)^\infty$ be the Fr\'echet--Lie algebra
of smooth vectors 
for the $\R$-action defined by $\alpha_t(x) = e^{t \bd} x e^{-t \bd}$. 
Then $\g^\infty$ is dense in $\fu_2(\sH)$ and we have a Fr\'echet--Lie group 
$\hat G^\infty := G^\infty \rtimes_\alpha \R$. 
If $\bd$  is diagonal, then 
\[ \fu(\sH)_{\rm alg} = \fu_2(\sH)_{\rm alg} \subeq \fu_2(\sH)^\infty.\]

In this subsection we show that, if $\hat\g^\infty$ contains open invariant
cones contained in the half spaces
\[ \g^\infty \oplus (0,\infty)\bd,\]
then the selfadjoint operator  $i\bd$ is semibounded. 
Intersecting with $\g^\infty + \bd \cong \g^\infty \times \{1\}$, 
we see that our requirement is equivalent to the existence of a 
proper open convex  subset of $\g^\infty$, invariant under the affine action 
\[ \beta_g x := \Ad_g x + g\bd g^{-1} - \bd\] 
on $\g^\infty$.

\begin{rem} Consider the half Lie group
  $G := \U_2(\sH) \rtimes_\alpha \R$,
  which is a topological group and a smooth manifold, such that all
  left multiplications are smooth
  (see \cite{MN18} for a detailed discussion of half Lie groups).
  Right multiplications with
  $g \in  \U_2(\sH)$ are of the form 
  $(h,t) \mapsto (h \alpha_t(g),t)$, hence are smooth if and only if
  $g \in \U_2(\sH)^\infty$, the subgroup of elements in
  $\U_2(\sH)$ with smooth
  $\alpha$-orbits. Hence 
  the Fr\'echet--Lie group $G^\infty = \U_2(\sH)^\infty \rtimes_\alpha \R$
  consists of all elements for which the right multiplication 
  is smooth (see~\cite{BHM23} for a systematic discussion of
  the Lie group structure on such subgroups). 
\end{rem}

\begin{lem} \mlabel{lem:semib}
  Let $A$ and $B$ be selfadjoint operators on the complex
  Hilbert space $\sH$ and assume that $B$ is bounded.
  Then $A$ is semibounded if and only if $A+B$ is.   
\end{lem}

\begin{prf} For unit vectors $v \in \cD(A) = \cD(A + B)$, we have
  \[ \la v,(A+B)v \ra
    = \la v, A v \ra + \la v, B v\ra
\in \la v, A v \ra + [-\|B\|, \|B\|].\] 
Thus $A$ is semibounded if and only if $A + B$ is.   
\end{prf}

\begin{thm} \mlabel{thm:5.8}
  If $\g^\infty = \fu_2(\sH)^\infty$ 
  contains proper open $\beta$-invariant convex subsets, 
then $i\bd$ is semibounded. 
\end{thm}

\begin{prf} Let $\Omega \subeq \g^\infty$ be a proper $\beta$-invariant open convex subset. 
Replacing $\bd$ by $\bd + x$ for some $x \in \g^\infty$, we may w.l.o.g.\ assume that 
$0 \in \partial \Omega$. Note that $i\bd$ is semibounded if and only if 
$i(\bd + x)$ has this property (Lemma~\ref{lem:semib}). 

The Hahn--Banach separation theorem implies the existence of a 
non-zero continuous linear functional 
$\lambda \in \g^{\infty,'}$ with $\inf \lambda(\Omega) = \lambda(0) = 0$. 
This implies that 
\[ \lambda(g\bd g^{-1}-\bd) \geq 0 \quad \mbox{ for } \quad g \in G^\infty.\] 
Derivatives along one-parameter subgroups of $G^\infty$ lead to the conditions 
\begin{equation}
  \label{eq:lamd}
 \lambda([x,\bd]) = 0 \quad \mbox{ and } \quad 
 \lambda([x,[x,\bd]]) \geq 0 \quad \mbox{ for } \quad x \in \g^\infty.
\end{equation}
The first condition implies that $\lambda \in (\g^\infty)'$ is invariant under 
the $\R$-action defined by conjugation with $e^{t\bd}$. 
For the Hilbert space $\g$, we consider the decomposition 
\[ \g^\infty = \g^0 \oplus (\g^\infty \cap (\g^0)^\bot).\] 
Then the definition of $\g^\infty$ implies that $\g^{\infty,0} = \g^0$, so that 
$\g^\infty  = \g^0 \oplus \oline{[\bd,\g^\infty]}$. Now $\lambda$ vanishes on $[\bd,\g^\infty]$, so that 
$\lambda \in (\g^{\infty,0})'$. Therefore 
$\lambda(x) = \tr(\bb_\lambda x)$ for some 
Hilbert--Schmidt operator $\bb_\lambda \in \g^{\infty,0}$. 

Next we evaluate the second condition. First we observe that 
$\bd$ commutes with $\bb_\lambda$, hence preserves all eigenspaces of 
$\bb_\lambda$. For the spectral decomposition of $i \bb_\lambda$, 
\[ \sH = \hat\bigoplus_{\mu \in \R} \sH^\mu, \quad \mbox{ we put } \quad
  \bd_\mu := \bd\res_{\sH^\mu}.\]
Then 
$\sH^\mu$ is finite-dimensional for $\mu\not=0$.
We claim that 
\[ h(x,y)
  := \lambda([x^*,[\bd,y]]) = - \lambda([x^*,[y,\bd]])\] 
defines a positive semidefinite hermitian form on $\g^\infty_\C$. In fact, the 
form is clearly sesquilinear, and from $\oline{\lambda(x)}
= \lambda(\oline x)$ and $\lambda([\bd,\g^\infty]) = \{0\}$ we derive 
\[ \oline{h(y,x)} 
= \lambda(\oline{[y^*,[\bd, x]]}) 
= \lambda([-y,[x^*,\bd]]) 
= \lambda([[x^*,\bd],y]) 
= \lambda([x^*,[\bd, y]]) = h(x,y),\]
so that $h$ is hermitian. We thus have 
\[ h(x,x) = \lambda([x^*,[\bd, x]]) = \lambda([[-x^*,[x,\bd]]) \geq 0
  \quad \mbox{ for all } \quad 
  x\in \g^\infty_\C\]
because this holds for $x \in \g^\infty$ by \eqref{eq:lamd}. 
Let $\nu\not=\mu$ and 
$x \in B(\sH^\mu, \sH^\nu)  = B_2(\sH^\mu, \sH^\nu) \subeq \g^\infty_\C$ 
(recall that at least one of these spectral subspaces is finite-dimensional). 
Then 
\[ [i\bb_\lambda,x] = (\nu-\mu) x \quad \mbox{ and } \quad
  [i\bb_\lambda,x^*] = (\mu-\nu) x^* \] 
lead to 
\begin{align*}
0 &\leq \lambda([x^*,[\bd, x]]) 
= \tr(\bb_\lambda [x^*,[\bd, x]])  \\
&= \tr(i\bb_\lambda [x^*,[-i\bd,x]])
= -\tr([i\bb_\lambda,x^*][i\bd, x])\\
&= -(\mu - \nu) \tr(x^*[i\bd,x])
= (\nu - \mu) \tr(i\bd xx^* - x^*x i\bd) \\
& = (\nu - \mu) \tr(i\bd_\nu xx^* - x^*x i\bd_\mu) \\
&= (\nu-\mu) \big(\tr(i\bd_\nu xx^*)- \tr(i\bd_\mu x^*x)\big).
\end{align*}

We assume w.l.o.g.\ that $\nu> \mu$. Consider an operators $x = P_{v,w}$ 
of the form 
$P_{v,w}(u) := \la w,u \ra v,$ 
where $w \in \sH^\mu$ and $v \in \sH^\nu$ are unit eigenvectors
for~$\bd$. Then 
\[ \tr(i\bd_\mu x^*x) 
= \tr(i\bd_\mu P_{w,w}) = \la w, i \bd_\mu w \ra 
\quad \mbox{ and } \quad 
\tr(i\bd_\nu xx^*) = \la v,i \bd_\nu v \ra.\] 
We therefore obtain 
\[ \sup \Spec(i\bd_\mu) \leq \inf\Spec(i\bd_\nu)\quad \mbox{ for } \quad
  \mu < \nu.\] 
Let us write this as $i\bd_\mu \leq i\bd_\nu$. 

If the Hilbert--Schmidt operator
$i\bb_\lambda$ has a positive eigenvalue, it has a maximal 
eigenvalue $\mu_{\rm max} > 0$. Then $i\bd_{\mu_{\rm max}} \geq i \bd_\mu$ for all 
other spectral values $\mu$ of $i\bb_\lambda$
and the fact that $\sH^{\mu_{\rm max}}$ is
finite-dimensional  implies that $i\bd$ is bounded from above. 
Likewise, the existence of a negative eigenvalue of $i\bb_\lambda$ 
implies that $i\bd$ is bounded from below. As $\lambda$ is non-zero, 
one of these two cases occurs, so that $i\bd$ is semibounded. 
\end{prf}

\begin{cor} \mlabel{cor:5.8}
  A semibounded smooth representation of the Fr\'echet--Lie group 
$\U_2(\sH)^\infty \rtimes_\alpha \R$ 
  exists if and only if the hermitian operator $i\bd$ on $\sH$ is semibounded. 
\end{cor}

\begin{prf} Theorem~\ref{thm:5.8} shows that the semiboundedness if
  $i\bd$ is necessary for the existence of a a semibounded smooth
  representation of $\U_2(\sH)^\infty \rtimes_\alpha \R.$

  To see that it is also sufficient, we observe that,
   if $i\bd$ is semibounded, then the identical representation
of $\U_2(\sH)$ on $\sH$ extends to a continuous 
representation $\hat\pi$ of $\U_2(\sH) \rtimes_\alpha \R$,
which defines a smooth representation of the Lie group
$\U_2(\sH)^\infty \rtimes_\alpha \R$. The corresponding space of
smooth vectors is the space $\sH^\infty = \sH^\infty(\bd)$ of smooth
vectors of the selfadjoint operator $i\bd$
(cf.\ \cite[Thm.~7.2]{Ne10a}). 
As the restriction of this representation to $\U_2(\sH)$ is bounded,
the cone $W_\pi$ contains an open half-space
$\fu_2(\sH) \pm \R_+ \bd$, so that $\hat\pi$ defines a semibounded
smooth representation of 
$\U_2(\sH)^\infty \rtimes_\alpha \R$.   
\end{prf}

\subsection{Representations of the trace class group} 
\mlabel{subsec:traceclassgroups}

From here on, in this section $\g = \fu_2(\sH), \g_1 = \fu_1(\sH),
G = \U_2(\sH)$ and $G_1 = \U_1(\sH)$.
We write $\ft_1 = \ft \cap \g_1$ for the subalgebra
of diagonal operators in $\fu_1(\sH)$, where $\ft \subeq \g$
also is the subalgebra of diagonal operator with respect to some
orthonormal basis. From \cite[Thms. I.11, III.8]{Ne98} and
\cite{Ne14} we have:

\begin{thm}
\mlabel{thm:7.1} 
For each $\lambda \in \cP_b$, there exists a unique 
bounded unitary representation $\rho_\lambda \: \g_1 \to \fu(\cH_\lambda)$ whose weight 
set with respect to $\ft_1$ is given by 
\[ \cP_\lambda = \conv(\cW\lambda) \cap (\lambda + \Z[\Delta]).\] 
Then $\cW\lambda = \Ext(\conv(\cP_\lambda))$ and 
$\rho_\lambda \sim  \rho_\mu \Longleftrightarrow \mu \in \cW\lambda.$ 
\end{thm}

\begin{defn} We call $(\rho_\lambda, \cH_\lambda)$ the {\it 
representation with extremal weight $\lambda \in \cP_b$}. 
\end{defn}

For more on convex sets and cones, we refer to Appendix~\ref{app:b}.

\begin{prop} \mlabel{prop:5.14}
For $\lambda \in \cP_b$, let $(\pi_\lambda, \cH_\lambda)$ be the 
corresponding norm continuous unitary representation of $G_1 = \U_1(\cH)$ 
and $I_\lambda \subeq \g_1'$ be the corresponding momentum set. 
Then 
\[ I_\lambda = I_\mu \qquad \Longleftrightarrow \qquad \oline{\conv\cW\lambda} 
= \oline{\conv\cW\mu},\] 
where closures refer to the weak-$*$-topology on
$i \ft_1'\cong \ell^\infty(J,\R)$.  
\end{prop}

\begin{prf} 
  By Fenchel duality, the weak-$*$-closed convex subsets
  $I_\lambda$ and $I_\mu$ are equal if and only if their support
  functionals $s_\lambda$ and $s_\mu$ coincide (\cite[\S 2.1]{Ne20}).
  As the support functionals $s_\lambda$ and $s_\mu$
  are continuous and $\Ad(G_1)$-invariant,
  and $\Ad(G_1)\ft_1$ is dense in $\g_1$, this is equivalent to 
  $s_\lambda\res_{\ft_1} = s_\mu\res_{\ft_1}$.
  These are the support functionals of the momentum sets
  $I_{\pi_\lambda\res_{T_1}}$ and $I_{\pi_\mu\res_{T_1}}$ of the restrictions to $T_1$.
  By Theorem~\ref{thm:7.1} these are the weak-$*$-closed convex
  hulls   $-i\oline{\conv(\cW\lambda)}$ and   $-i\oline{\conv(\cW\mu)}$
  in $\ft_1'$,   respectively.
\end{prf}

As $\rho_\lambda \cong \rho_\mu$ if and only if
$\mu \in \cW\lambda$ (\cite[Thm.~5.12]{Ne04}), the momentum set
$I_\lambda$ is a rather coarse invariant. It is easy to see examples 
of non-equivalent representations with the same momentum
set. A~simple example is obtained for 
\[ J = \Z, \qquad  \lambda = \chi_{\N}\quad \mbox{ and } \quad
  \mu = \chi_{\N_0}\]
(see also Remark~\ref{rem:5.15}).

We recall the following result from  \cite[Thm.~VII.18]{Ne04}.
  For further details we refer to \cite{Ne98, NeSt01}.

\begin{thm} \mlabel{thm:ext-res}
  {\rm(Extension Theorem)}
  Let $\sH = \ell^2(J,\C)$, let  $\lambda : J \to \Z$  be a bounded weight
and $(\rho_\lambda, \cH_\lambda)$ be the corresponding
bounded unitary representation of $G_1 = \U_1(\sH)$. For the derivation
of $\fu_1(\sH)$, defined by
\[ \bd = \diag(i\lambda_j), \quad \mbox{ and } \quad
  \alpha_t(g) = e^{t\bd} g e^{-t\bd}\]
we consider the restricted group
\[ \U_{\rm res}(\sH,\bd) := \{ g \in \U(\sH) \:
  g \bd g^{-1} - \bd \in \fu_2(\sH)\}, \]
endowed with its natural Banach--Lie group structure. 
Then $\rho_\lambda$ extends to a projective representation
of the group $\U_{\rm res}(\sH,\bd)$ on $\cH_\lambda$.
In particular, it restricts to a projective representation
of $\U_2(\sH)$.
\end{thm}

Here the representations $\rho_\lambda$ come from
Theorem~\ref{thm:7.1} and the construction of the 
Lie group structure on $\U_{\rm res}(\sH,d)$ can be found in
\cite{Ne04}. Similar to the discussion of central extensions in
Subsection~\ref{subsec:centext}, it
is based on the observation that we have a surjective
homomorphism
\[ \U_{1,2}(\sH,\bd) \rtimes \U(\sH)^0 \to \U_{\rm res}(\sH,\bd),\quad
  (g_1, g_2) \mapsto g_1 g_2.\]
Here
\[ \U_{1,2}(\sH,\bd)
  = \{  g \in \U_2(\sH) \: (\forall \ell = 1,\ldots,k)\
  \|g_{\ell\ell}-\1\|_1 < \infty \},\]
where $g = (g_{\ell,m})_{1 \leq \ell,m \leq N}$ is the matrix
description of $g$ with respect to the different eigenspaces 
$\sH_m := \ker(\bd - i d_m\1)$, $m = 1,\ldots, N$, of $i\bd$.
Accordingly
\[ \U(\sH)^0 \cong \prod_{m = 1}^N \U(\sH_m).\]

The construction of the
extension of the representation $\rho_\lambda$ of $\U_1(\sH)$ to
$\U_{\rm res}(\sH,\bd)$ in \cite{Ne04} shows that, on the subgroup
$\U_{1,2}(\sH,\bd)$, in which $\U_1(\cH)$ is dense,
it extends by continuity to a unitary representation.
Moreover, the $\lambda$-weight space is one-dimensional, of the form
$\C \Omega_\lambda$, and $\Omega_\lambda$ is fixed by the subgroup
$\U(\sH)^0$.

Theorem~\ref{lem:7.7} below shows the extension of
$\rho_\lambda$ to a projective representation of
$\U_{\rm res}(\sH,\bd)$ is maximal. 

We also note that a Lie algebra cocycle corresponding to the
central extension $\fu_2(\cH)^\sharp$, defined by this projective
representation, can be evaluated with the methods explained in
Example~\ref{ex:4.13}. This leads to
\begin{equation}
  \label{eq:om-lam}
  \omega(x,y) = i\lambda([x,y]) = \omega_\bd(x,y),
\end{equation}
if we consider $\lambda \in \ell^\infty(J,\R)$ as a
continuous linear functional on $\fu_1(\sH)$. 

\subsection{Covariance of bounded representations of
  the trace class ideal}
\mlabel{subsec:5.4} 

Let $\lambda \: J \to \R$ be bounded and integral, i.e., 
$\lambda_j - \lambda_k \in \Z$ for $j,k\in J$, and 
$(\rho_\lambda, \cH_\lambda)$ be the corresponding bounded 
unitary representation of the simply connected covering group 
$\tilde\U_1(\sH)$ (Theorem~\ref{thm:7.1} and \cite{Ne98}).  
In this subsection we determine the covariance group 
\[ \U(\sH)_\lambda := \{ g \in \U(\sH) \: \rho_\lambda \circ \alpha_g \sim 
\rho_\lambda\},\] 
where $\U(\sH)$ acts by the automorphisms $\alpha_g(h) := ghg^{-1}$ 
on $\U_1(\sH)$, and also  by canonical lifts to~$\tilde\U_1(\sH)$. 

\begin{ex} (a) If $\lambda$ is constant $c$, then 
$\rho_\lambda(g) = \det^\lambda(g)$, as a representation of $\tilde\U_1(\sH)$ 
if $c\not\in \Z$. This one-dimensional representation is $\U(\sH)$-invariant, 
so that $\U(\sH)_\lambda = \U(\sH)$ in this case. 

\nin (b) If $\lambda$ is finitely supported, then $\rho_\lambda$ extends to a 
bounded unitary representation of $\U(\sH)$, so that we also obtain in this 
case $\U(\sH)_\lambda = \U(\sH)$ (\cite{Ne98, Ne14b}). 

\nin (c) If there exists a constant $c$ such that $\lambda - c$ is finitely supported, 
then we combine (a) and (b) to see that $\U(\sH)_\lambda = \U(\sH)$ also holds 
in this case. 
\end{ex}

\begin{rem} \mlabel{rem:5.15}
  Let $(e_j)_{j \in J}$ be an orthonormal basis of $\sH$ and 
  $g \in \U(\sH)$ be a permutation of the basis elements. We write
  $S_J \subeq \U(\sH)$ for the subgroup of all unitary operators
  mapping the orthonormal basis onto itself and identify it
with the group of all bijections 
    of the set $J$; we write $S_{(J)} \trile S_J$ for the normal subgroup
    of those permutations moving only finitely many elements.

Since $\rho_\lambda$ is uniquely determined by the weight set 
$\cP_\lambda \subeq \hat T_1$ for the diagonal torus $T_1\subeq \U_1(\sH)$ 
(\cite{Ne98}), we have $g \in \U(\sH)_\lambda$ if and only if 
$g\lambda \in \cW\lambda = \Ext(\conv(\cP_\lambda))$
(note that $\cW \cong S_{(J)}$).
This is always the case if $\supp(\lambda)$ is finite or if
there exists a constant $c \in \R$, for which $\lambda - c$ has finite support.
If this is not the case, then $\lambda$ takes two different values infinitely
often and there exists a $g \in S_J$ with
$g.\lambda\not\in \cW\lambda$; see also Remark~\ref{rem:7.6} below. 
\end{rem}

\begin{rem} \mlabel{rem:7.6}
  A necessary condition for $g \in \U(\sH)_\lambda$ is that
  the momentum set $I_{\rho_\lambda}$ satisfies 
  $\alpha_g^*I_{\rho_\lambda} = I_{\rho_\lambda}$ (cf.\ Proposition~\ref{prop:5.14}),
  resp., that 
  \begin{equation}
    \label{eq:slambda-cov}
    s_{\lambda} \circ \alpha_g =   s_{\lambda}
  \end{equation}
holds for the support functional $s_\lambda := s_{\rho_\lambda} \:  \fu_1(\sH) \to \R$. 
The following example shows that \eqref{eq:slambda-cov} does not
imply $g \in \U(\sH)_\lambda$. 

If $\lambda = \chi_M$ is the characteristic function of a subset 
$M \subeq J$ which is infinite and co-infinite, then any
$g \in S_J$ satisfies $g\lambda = \chi_{g.M}$ 
and 
\[ \cW\chi_M = \{ \chi_E \: |E \setminus M| = |M \setminus E| < \infty \} \] 
imply that $g \in \U(\sH)_\lambda$ is equivalent to 
\[ |gM \setminus M| = |M \setminus gM| < \infty.\]

We now calculate the support functional. 
For $x \in i\ft_1\cong \ell^1(J,\R)$, we have 
\[  s_\lambda(-ix) = \sup \la \cW\chi_M, x \ra 
 = \sup \Big\{ \sum_{j \in E} x_j \: |E \setminus M| = |M \setminus E| < \infty \Big\}.\] 
If $\supp(x)$ is finite and $x = x^+ - x^-$ with 
$x^\pm := (\pm x) \vee 0$ is the canonical 
decomposition, then 
\[  s_\lambda(-ix) = \sum_{j \in J} x^+_j
  \quad \mbox{ for } \quad x \in i \ft_1 \cong \ell^1(J,\R).\]
By continuity of the map $x \mapsto x^+$ and of $s_\lambda$, it follows that 
$s_\lambda(-ix) = \tr(x^+)$ for $x \in i\ft,$
and by conjugation invariance of $s_\lambda$, we further get
\begin{equation}
  \label{eq:slambda}
  s_\lambda(-ix) = \tr(x^+) \quad \mbox{ for } \quad x \in \Herm_1(\sH).
  \end{equation}

This is true for every infinite $M$ which is also co-infinite. 
Since there are many such sets which are not conjugate under the Weyl group 
$\cW \cong S_{(J)}$, it follows that the momentum set $I_{\rho_\lambda}$ 
does not determine the representation up to equivalence.
Concretely, the weak-$*$-closure of $\conv(\cW\chi_M)$
does not distinguish subsets $M_1$ and $M_2$ of the same cardinality
if the complements $M_1^c$ and $M_2^c$ also have the same cardinality, i.e.,
if $\chi_{M_2} \in S_J.\chi_{M_1}$. 
Typical examples of such sets are 
$M_1 = \N \subeq J = \Z$ and $M_2 = \N_0$. 
\end{rem}

\begin{thm}  \mlabel{lem:7.7}
  Suppose that $\lambda$ is bounded and
  define a diagonal operator $\bd \in \fu(\sH)$ by 
$\bd e_j = i \lambda_j e_j$ for~$j \in J$. 
Then
\[ \U(\sH)_\lambda = 
  \U_{\rm res}(\sH, \bd) := \{ g \in \U(\sH) \: \|g\bd g^{-1}-\bd\|_2 < \infty \}.\] 
\end{thm}

\begin{prf} We know already from
  Theorem~\ref{thm:ext-res}
  that ``$\supeq$'' holds. So it remains to verify
  ``$\subeq$''.
  Recall from Theorem~\ref{thm:ext-res} 
  that $\rho_\lambda$ extends to a smooth projective 
unitary representation 
$\oline\rho_\lambda \: \U_2(\sH) \to \PU(\cH_\lambda)$ 
(see \cite{JN15} for the concept of a smooth projective unitary representation). 

If $g \in \U(\sH)_\lambda$, then the density of 
$\U_1(\sH)$ in $\U_2(\sH)$ implies that 
the projective unitary representation $\oline\rho_\lambda \circ \alpha_g$ 
is equivalent to $\oline\rho_\lambda$, and this implies in 
particular that the corresponding central extensions 
are equivalent. On the Lie algebra level a cocycle 
of the central extension is given by 
\[ \omega_\lambda(x,y) = i\lambda([x,y]) = \tr(\bd [x,y])
\quad \mbox{ for }  \quad x,y \in \fu_2(\sH) \] 
(see \eqref{eq:om-lam} in Subsection~\ref{subsec:traceclassgroups}). We thus obtain the necessary condition 
\[ g.[\omega_\lambda] = [\omega_\lambda] \quad \mbox{ in } \quad H^2(\fu_2(\sH),\R).\] 
From 
\[ (g.\omega_\lambda)(x,y) 
= \omega_\lambda(\alpha_g^{-1}x, \alpha_g^{-1}y)
= \tr(\bd \alpha_g^{-1}[x,y])
= \tr(\alpha_g(\bd)[x,y])\] 
we derive that the operator 
$g\bd g^{-1} - \bd$ has to represent a trivial cocycle, which means that 
\[ g\bd g^{-1} - \bd \in \R i\1 + \fu_2(\sH)\] 
(\cite[Ex.~2.18]{Ne02c} and Subsection~\ref{subsec:centext}). 
We consider the $1$-cocycle
\[ \beta \: \U(\sH) \to \fu(\sH), \quad
\beta(g) := g \bd g^{-1} - \bd.\] Suppose that 
$\beta(g) - \mu i\1 \in \fu_2(\sH).$ 
Let $\oline\beta \: \U(\sH) \to \fu(\sH)/\fu_\infty(\sH)$ be the corresponding 
cocycle modulo the subspace $\fu_\infty(\sH)$ of compact skew-hermitian
operators. 
Then $\oline\beta(g) = \mu [i\1]$, 
and since $[i\1]$ is fixed, we obtain 
\[ \oline\beta(g^n) = n\mu [i\1] \quad \mbox{ for } \quad n \in \Z.\] 
As $\oline\beta$ is bounded, it follows that $\mu = 0$, hence that 
$\beta(g)$ is Hilbert--Schmidt, 
which is equivalent to 
$g\bd - \bd g = (g\bd g^{-1} - \bd)g \in B_2(\sH)$, i.e., to 
$g \in \U_{\rm res}(\sH,\bd)$.
\end{prf} 

Note that the group $\U_{\rm res}(\sH, \bd)$ does not change in the theorem above, if we replace $\bd$ by $\bd + \ad x$ for some $x \in \fu_2(\sH)$
(see Corollary~\ref{cor:4.6}).

\begin{rem} Assume that $\bd$ is diagonalizable with finitely
  many eigenvalues and eigenspaces
  $\sH_j := \ker(\bd - i d_j \1)$, $j = 1,\ldots, r$.
Accordingly, we write $g \in \U(\sH)$ as a matrix 
$g = (g_{ij})_{1 \leq i,j \leq r}$ with 
$g_{ij} \in B(\sH_j, \sH_i)$. The relation 
$g\bd - \bd g\in B_2(\sH)$ means that, for 
$i\not=j$, the corresponding matrix entry 
$(d_j - d_i)g_{ij}$ is Hilbert--Schmidt, and this is equivalent to 
$g_{ij}$ being Hilbert--Schmidt. In this sense we have 
\[ \U_{\rm res}(\sH,\bd) = \{  g = (g_{ij}) \: (\forall i \not=j)\, g_{ij} \in 
B_2(\sH_j, \sH_i)\}.\] 
\end{rem}

\section{Perspectives} 
\mlabel{sec:6}

Although we formulated several challenging problems in the
main text, we briefly discuss in this short final section
some directions of research that target important tools to
analyze unitary representations of infinite-dimensional Lie groups. 

\subsection{Representations of abelian Hilbert--Lie groups}

\begin{defn} A {\it toroidal Hilbert--Lie group} is a group of the form
$T = \ft/\Gamma$, where $\Gamma$ is a discrete subgroup of 
the real Hilbert space $\ft$, considered as an abelian Hilbert--Lie algebra
(\cite{Ba91}), and $\Gamma$ spans a dense subspace of~$\ft$.
Assigning to $\alpha \in \ft'$ with $\alpha(\Gamma) \subeq \Z$
the character $\chi_\alpha(x + \Gamma) := e^{2\pi i \alpha(x)}$, we then obtain
an isomorphism of groups
\[  \hat T \cong \{ \alpha \in \ft' \: \alpha(\Gamma) \subeq \Z \}
  \into \Hom(\Gamma,\Z).\] 
\end{defn}

Any toroidal Lie group contains the torus groups 
$T_F := (\Spann_\R \Gamma_F)/\Gamma_F$, $\Gamma_F \subeq \Gamma$
finitely generated, and the directed union $T_0$ of these subgroups
is dense in $T$. Assume that $\Gamma$ is countable.
Then $T_0$ carries an intrinsic Lie group structure whose underlying
topology is the direct limit topology (\cite{Gl05}). We
identify the character group $\hat T$ with a subgroup  of
$\hat{T_0} \cong \Hom(\Gamma,\Z)$, consisting of
those characters of   $T_0$ which extend continuously to $T$.

An important example is the group  
$T := \ell^2/\Z^{(\N)}$ with
the dense subgroup $T_0 := \R^{(\N)}/\Z^{(\N)}$
and $\hat{T_0} \cong \Z^\N \supeq \hat T \cong \Z^{(\N)}$.
That bounded representations of $T$ are direct sums of eigenspaces follows
from Theorem~\ref{thm:z-discrete} because $\hat T$ is countable. 

It is an interesting problem to understand the
continuous unitary representations of $T$ in terms of their
restrictions to $T_0$. This problem immediately reduces to cyclic
representations, which are given by positive definite functions. 
As $T_0$ is a nuclear group, positive definite functions are Fourier
transforms of measures $\mu$ on $\hat{T_0}$ (\cite{Ba91}).

\begin{prob}  Characterize those Borel measures $\mu$ on $\hat{T_0}
  \cong \Hom(\Gamma,\Z)$ whose
  Fourier transform $\hat\mu$ extends to a continuous function on $T$. 
  This is equivalent to the function
  $\hat\mu$ to extend from $\Spann \Gamma$ to a continuous
  function on $\ft$.
  As $\hat T$ is a countable subset of $\hat{T_0}$ and $\mu$ is atomic on
  $\hat{T}$, it suffices to consider measures on the complement
  $\hat{T_0} \setminus \hat T$. 
\end{prob}

We expect that this problem can be addressed by using the
Sazonov topology on Hilbert spaces; see \cite{Ya85} and \cite{Xia72}
for tools on measures on infinite-dimensional spaces. 

\subsection{Representations of non-Lorentzian double extensions}

The construction in Section~\ref{subsec:5.4}
produces in particular projective semibounded
representations of groups  $G \rtimes_\alpha \R$, $G = \U_2(\sH)$, 
where $\alpha$ is defined by a bounded diagonal operator $\bd$ on $\sH$ 
with integral eigenvalues and the cocycle $\omega$ of the corresponding
central Lie algebra extension $\g^\sharp = \R \oplus_\omega \g$
is also given by $\bd$.

In general, one may consider central extensions whose cocycle
is given by a bounded operator $\bd_0$, and $\alpha$ is generated
by a, possibly unbounded diagonal operator $\bd$. It would be interesting
to understand the precise relations on the pair $(\bd_0, \bd)$
that characterizes the existence of projective semibounded
representations of $G \rtimes_\alpha \R$.

For extremal weight representations $(\pi_\lambda,\cH_\lambda)$
as in Theorem~\ref{thm:ext-res}, 
this condition only depends on the class of $\bd$ modulo
$\ell^1$ if the weight $\lambda$ is bounded; see \cite{MN16}.
If the spectrum of $\bd$ is bounded from below,
a  suitable diagonal modification $\tilde\bd$ of $\bd$
should have a non-zero minimal eigenspace.

\subsection{Restrictions imposed by spectral conditions}

Suppose that $(U,\cH)$ is a unitary representation
of $G = \U_2(\sH)$ and $\alpha_t(g) = e^{t\bd} g e^{-t\bd}$  is a
one-parameter group of automorphisms for which $U$ extends
to a unitary representation $\hat U$ of $G \rtimes_\alpha \R$ for which
the one-parameter group $V_t = U(e,t)$ has semibounded spectrum.
In view of Corollary~\ref{cor:5.8}, we only expect this to happen
if $i\bd$ is semibounded. 

\begin{prob} Determine the restrictions on the representation
  $(U,\cH)$ of $G = \U_2(\sH)$ imposed by the semiboundedness of~$V$.
\end{prob}

In view of Problem~\ref{prob:4.22},
we do not expect any conclusive
restrictions without  semiboundedness of~$V$.

\newpage

\appendix

\section{Invariant convex subsets of Hilbert--Lie algebras}

\begin{lem} \mlabel{lem:1.3}  {\rm(\cite[Lemma~A.4]{Ne20})} 
Let $(V,\|\cdot\|)$ be a normed space 
and $U \subeq V$ be a proper open convex subset with 
complement $U^c = V \setminus U$. 
Then the distance function 
\[ d_U(x) := \dist(x,U^c) 
= \inf \{ \|x-y\| \: y \in U^c \}\]  
is continuous,  concave and 
$d_U^{-1} \: U \to \R$ 
is a continuous convex function on $U$ with 
\[ \lim_{x_n \to x \in \partial U} d_U^{-1}(x_n) = \infty.\] 
In particular, all the sets 
\[ U_c := \{ x \in U \: d_U(x) \geq c \}, \quad c > 0, \] 
are closed subsets of $V$. 
\end{lem} 

\begin{prop} \mlabel{prop:2.12} If $G$ is a connected Hilbert--Lie group, 
then each non-empty open $\Ad(G)$-invariant 
convex subset $U \subeq \g$ intersects the center. 
\end{prop}

\begin{prf} If $U = \g$, there is nothing to show because $0 \in U$. 
We may therefore assume that $U \not=\g$. 
Then $U$ is a proper open convex subset and the group 
$\Ad(G)$ acts isometrically on 
$\g$ preserving $U$. Let $d_U \: V \to \R$ denote the distance from 
$\partial U$. 
Then Lemma~\ref{lem:1.3} shows that, for every 
$c > 0$, 
\[  U_c := \{ x \in U \: d_U(x) \geq c \} \] 
is a closed convex subset of $\g$. Let $c> 0$ be such that $U_c \not=\eset$. 
All these sets are $\Ad(G)$-invariant. Now $U_c$ is a Bruhat--Tits space 
with respect to the induced Hilbert metric from $\g$ and 
since $\Ad(G)$ acts on this space isometrically with bounded 
orbits, the Bruhat--Tits Fixed Point Theorem 
implies the existence of a fixed point $z \in U_c\subeq U$ (\cite{La99}). 
Now it only remain to observe that 
$\z(\g)$ is the set of fixed points for the action of 
$\Ad(G)$ on $\g$. 
\end{prf}

\newpage

\section{Central extensions and convexity}
\mlabel{app:b}

\begin{definition}   \mlabel{def:3.1}
(a) Let $V$ be a locally convex space. We consider the {\it
affine group} $\Aff(V) \cong V \rtimes \GL(V)$ which acts on $V$ by 
$(x,g).v = g.v + x$. On the space $\tilde V := V \times \R$ the group 
$\Aff(V)$ acts by linear maps 
\[ (x,g).(v,z) := (g.v + z x, z),\] and we
thus obtain a realization of $\Aff(V)$ as a subgroup of $\GL(\tilde V)$. 
The corresponding Lie algebra is 
$\aff(V) \cong V \rtimes \gl(V)$ with the bracket 
$$ [(v,A),(v',A')] = (Av' - A'v, [A,A']). $$

\nin (b) A homomorphism $\rho \: \g \to \aff(V)$ 
is therefore given by a pair 
$(\rho_l, \theta)$, consisting of a linear representation 
$\rho_l \: \g \to \gl(V)$ and map $\theta \: \g \to V$ satisfying 
\begin{equation}
  \label{eq:6.1}
\theta([x,y]) = \rho_l(x)\theta(y) -\rho_l(y)\theta(x) 
\quad \hbox{ for } \quad x,y \in \g. 
\end{equation}
A linear map $\theta \: \g \to V$ satisfying \eqref{eq:6.1} is called a 
{\it $1$-cocycle} with values in the representation $(\rho_l,V)$ of~$\g$. 
We write $Z^1(\g,V)_{\rho_l}$ for the set of all such cocycles. 

\nin (c) On the group level, a homomorphism $\rho \: G \to \Aff(V)$
is given by a pair 
$(\rho_l, \gamma)$ of a linear representation 
$\rho_l \: G \to \GL(V)$ and map $\gamma \: G \to V$ satisfying 
\begin{equation}
  \label{eq:6.2}
\gamma(g_1 g_2) = \gamma(g_1) + \rho_l(g_1)\gamma(g_2) 
\quad \hbox{ for } \quad g_1, g_2 \in G.
\end{equation}
A map $\gamma \: G \to V$ satisfying \eqref{eq:6.2} is called a 
{\it $1$-cocycle} with values in the representation $(\rho_l,V)$ of~$G$. 
We write $Z^1(G,V)_{\rho_l}$ for the set of all such smooth cocycles. 
Typical examples of $1$-cocycles are the {\it coboundaries} 
$\gamma(g) := \rho_l(g)v - v$ for some $v \in V.$
\end{definition}

\begin{proposition} \mlabel{prop:3.3} Let $G$ be a connected 
Lie group with Lie algebra $\g$ and 
$\omega \in Z^2(\g,\R)$ be a continuous
$2$-cocycle. Then the following assertions hold 
\begin{itemize}
\item[\rm(a)] $\theta_\omega(x)(y) := \omega(x,y)$ is a $1$-cocycle 
on $\g$ with values in the coadjoint representation $(\ad^*,\g')$, where 
$\ad^*(x)\beta = - \beta \circ \ad x$.
\item[\rm(b)] If $G$ is simply connected and 
$\hat\g = \R \oplus_\omega \g$ the central extension 
defined by $\omega$, then there exists a unique smooth representation 
\begin{equation}
  \label{eq:adhatact}
\Ad_{\hat\g} \: G \to \Aut(\hat\g), \quad 
\Ad_{\hat\g}(g)(z,y) = (z + \Theta_\omega(g^{-1})(y), \Ad(g)y) 
\end{equation}
whose derived representation is given by 
$$ \ad_{\hat\g}(x)(z,y) = (\omega(x,y),[x,y]) = [(0,x),(z,y)], $$
so that  $T_e(\Theta_\omega)(x) = -\theta_\omega(x) = \omega(\cdot, x)$. 
The corresponding dual representation of $G$ on 
$\hat\g' \cong \R \times \g'$ is given by 
$$ \Ad^*_{\hat\g}(g)(z,\alpha) = (z, \Ad^*(g)\alpha + z\Theta_\omega(g)). $$
\item[\rm(c)] Identifying $\g'$ with the hyperplane $\{1\} \times \g'$ in
$\hat\g'$, we thus obtain an affine action 
$$\Ad^*_\omega(g)\alpha = \Ad^*(g)\alpha + \Theta_\omega(g) $$ 
of $G$ on $\g'$. 
\end{itemize}
\end{proposition}

\begin{proof} (a) follows from an easy calculation 
(\cite[Lemma~VI.4]{Ne04}). 

\nin (b) The uniqueness of the affine representation follows from 
the general fact that for a connected Lie group, smooth representations 
are uniquely determined by their derived representations 
(\cite[Rem.~II.3.7]{Ne06}). 
The existence follows form \cite[Prop.~7.6]{Ne02} because 
$\R$ is complete. In loc.\ cit.\ one finds more details on 
how to obtain the corresponding cocycle: 
For each $x \in \g$, let $f_x \in C^\infty(G,\R)$ 
be the unique function with $\dd f_x = i_{x_r}\Omega$ and $f_x(e) = 0$,  
where $\Omega \in \Omega^2(G,\R)$ is the left invariant $2$-form with $\Omega_e = 
\omega$ and $x_r(g) = x.g$ is the right invariant vector field with $x_r(e) = x$. 
Then 
$\Theta_\omega(g)(y) := f_y(g)$ defines a smooth function $G \times \g \to \R$: 
\[ \Theta_\omega(g)(y) 
= \int_0^1 \Omega_{\gamma_g(t)}\big(y_r(\gamma_g(t)), \gamma_g'(t)\big)\, dt
= \int_0^1\omega\big(\Ad(\gamma_g(t))^{-1}y, \gamma_g'(t)\big)\, dt, \]
where $\gamma_g$ is any piecewise smooth path in $G$ from $e$ to $g$. 
Note that 
$$ \big(T_e(\Theta_\omega)(x)\big)(y) = (i_{y_r}\Omega)_e(x) 
= \omega(y,x) = -\theta_\omega(x)(y). $$

\nin (c) is an immediate consequence of (b). 
\end{proof}

\begin{defn} We call a Banach--Lie algebra {\it elliptic}
  if the adjoint group $\Inn(\g) := \la e^{\ad \g} \ra$ is bounded.   
\end{defn}

\begin{cor} \mlabel{cor:b.3} Suppose that $\g$ is a Banach--Lie algebra 
and $\hat\g = \R \oplus_\omega \g$ is the central extension defined by 
$\omega \in Z^2(\g,\R)$. Then $\hat\g$ is elliptic
if and only if $\g$ is elliptic and $\Theta_\omega$ is bounded.
\end{cor}

\begin{proof} From \eqref{eq:adhatact} is follows that 
the adjoint group $\la e^{\ad\hat\g}\ra = \la e^{\ad_{\hat\g} \g}\ra$ of 
$\hat\g$ is bounded if and only if the adjoint group of $\g$ 
is bounded and $\Theta_\omega(G)$ is a bounded subset of $\g'$. 
\end{proof}

The following theorem extends some techniques developed
in \cite[\S 3]{NZ13} for abelian Hilbert--Lie algebras
(for which the double extensions are oscillator algebras)
to the case of simple Hilbert--Lie algebras. 

\begin{defn} \mlabel{def:seq} Let $V$ be a locally convex space. 
    We call a subset $X \subeq V'$
    {\it semi-equicontinuous} if its {\it support functional} 
\[  s_X \: V \to \R \cup \{\infty\}, \quad s_X(v) 
= - \inf \la X,v \ra = \sup \la X,-v \ra  \] 
is bounded on some non-empty open subset of $V$.
\end{defn}

\begin{thm} \mlabel{thm:b.4} Let $\fg$ be a Hilbert--Lie algebra, 
$G$ a corresponding simply connected Lie group, 
$\omega \in Z^2(\g,\R)$ be a continuous cocycle,  
and $\Theta_\omega \: G \to \g'$ the corresponding 
$1$-cocycle with $T_e(\Theta_\omega)x =  \omega(\cdot, x)$ for $x \in \g$. 
On $\g'$ we consider the affine action $g * \nu := \Ad^*(g)\nu + \Theta_\omega(g)$. 
Then the following are equivalent: 
\begin{itemize}
  \item[\rm(a)] $\omega$ is a coboundary. 
  \item[\rm(b)] $\Theta_\omega$ is bounded.  
  \item[\rm(c)] $\Theta_\omega(G)$ is semi-equicontinuous. 
  \item[\rm(d)] One orbit of the affine action of $G$ on $\g'$ is semi-equicontinuous. 
  \item[\rm(e)] All orbits of the affine action of $G$ on $\g'$ are bounded.  
  \end{itemize}
\end{thm}

\begin{prf} (a) $\Rarrow$ (b): If $\omega$ is a coboundary, 
then there exists an $\alpha \in \g'$ with 
$\omega(x,y) = \alpha([y,x])$ for $x,y \in \g$ 
and $\Theta_\omega = \alpha - \Ad^*(g)\alpha$ 
follows from the equality of derivatives 
\[ (\ad^* x) \alpha = -\alpha([x,\cdot]) = \omega(x,\cdot) 
= -T_e(\Theta_\omega)(x).\] 
Since the coadjoint action preserves the norm on $\g'$, 
the boundedness of $\Theta_\omega$ follows. 

\nin (b) $\Rarrow$ (c) is trivial. 

\nin (c) $\Leftrightarrow$ (d) follows from the fact that all orbits 
of the coadjoint action of $G$ on $\g'$ are bounded. 

\nin (b) $\Leftrightarrow$ (e) follows with the same argument. 

\nin (c) $\Rarrow$ (b): Since $Z(G)_0$ acts trivially on $\g$, 
we obtain for $z \in \fz := \z(\g)$ the relation
\[ \Theta_\omega(\exp z) = \omega(\cdot,z).\] 
Therefore $\Theta_\omega(\exp \fz)$ is a linear space. As it is 
semi-equicontinuous it is trivial. This means that ${\omega(\fz,\g) = \{0\}}$. 
Accordingly, $\Theta_\omega(G) \subeq \fz^\bot \cong [\g,\g]'$. 

The set $W$ of all points $x_0 \in \g$ for which the support functional 
$s_{\Theta_\omega(G)}(x) := \sup \Theta_\omega(G)(-x)$ 
is bounded in a neighborhood of $x_0$ is a non-empty 
open invariant convex cone. We have seen in the preceding 
paragraph that $W + \fz = W$. Hence $W \cap [\g,\g] \not=\eset$, 
and therefore $[\g,\g] \subeq W$ follows from 
Proposition~\ref{prop:2.12}. Therefore $W = \g$, which 
implies that $\Theta_\omega(G)$ is equicontinuous, i.e., 
that $\Theta_\omega$ is bounded. 

\nin (b) $\Rarrow$ (a): If $\Theta_\omega$ is bounded, then the orbit of $0 \in \g'$ 
under the affine action defined by $\Theta_\omega$ is bounded. 
Hence the Bruhat--Tits 
Theorem (\cite{La99}) implies the existence of a fixed point $\alpha$ of the corresponding 
affine isometric action on $\g'$. This means that 
$\Theta_\omega(g) = \alpha - \Ad^*(g)\alpha$ for each $g\in G$, 
and by taking derivatives in $e$, we see that 
$\omega(y,x) = \alpha([x,y])$ is a coboundary. 
\end{prf}

\begin{rem}
  \mlabel{rem:new1}
  Let $\g^\sharp = \R \oplus_\omega \g$ be the central Lie algebra
  extension corresponding to the cocycle $\omega \in Z^2(\g,\R)$.
  Then we have a natural affine embedding
  $\g' \into (\g^\sharp)', \alpha \mapsto (1,\alpha)$,
  and this embedding intertwines the coadjoint action of the
  group $G$ in Theorem~\ref{thm:b.4} with the affine action
  on $\g'$, specified by $(g,\alpha) \mapsto g * \alpha$. Therefore
  coadjoint orbits of $\g^\sharp$ not vanishing on the central element
  $(1,0)$ correspond to affine orbits in~$\g'$.   
\end{rem}

\begin{ex} \mlabel{ex:hs} Let $\sH$ be a complex Hilbert space and 
$\g = \fu_2(\sH)$. 

\nin (a) As we have seen in
Subsection~\ref{subsec:centext},
each cocycle $\omega \in Z^2(\g,\R)$ can be written as 
$$ \omega(x,y) = \tr([\bd,x]y) = \tr(\bd[x,y]) $$
for some bounded skew hermitian operator $\bd \in \fu(\sH)$ (\cite[Prop.III.19 \& proof]{Ne03}). Then 
$\theta_\omega(x) = [\bd,x]$ if we identify $\g$ and $\g'$ via the 
invariant scalar product defined by $(x,y) := -\tr(xy) = \tr(xy^*)$. 
From that we easily derive that $\Theta_\omega \: G = \U_2(\sH) \to \g$ 
is given by 
$$ \Theta_\omega(g) = g\bd g^{-1}- \bd. $$
This is a cocycle for the adjoint action, which is an orthogonal 
representation of $G$ on the real Hilbert space $\g$. 
In view of Theorem~\ref{thm:b.4}, 
$\Theta_\omega$ is bounded if and only if $\omega$ is a coboundary, which is equivalent to 
$\bd \in \R i \1 + \fu_2(\sH)$. 
In particular, the central extension defined by a 
non-trivial $2$-cocycle $\omega$ is not an elliptic Banach--Lie algebra,
hence in particular not Hilbert (Corollary~\ref{cor:b.3}). 

\nin (b) The existence of unbounded cocycles $\theta$ for the adjoint action 
of $\U_2(\sH)$ also implies that we obtain interesting representations 
of  this group by composing the Fock representations 
of the affine group $\Mot(\fu_2(\sH)) \cong \fu_2(\sH) \rtimes \OO(\fu_2(\sH))$ 
with the homomorphism $(\theta,\id) \: \UU_2(\sH) \to \Mot(\fu_2(\sH))$. 
For details we refer to \cite{Gu72}. 
\end{ex}

\begin{remark} We may now ask for the convexity properties of the affine orbit 
${\cal O}_\omega = \Theta_\omega(G)$, resp., the coadjoint orbit 
of $(1,0)$ in $\hat\g'$. Clearly, the convexity properties of this 
orbit will depend on the cocycle $\omega$ and 
not only on its cohomology class. F.i.~if the central 
extension is trivial, then the affine action on $\g'$ is equivalent 
to the coadjoint action of $\g'$ and different orbits have 
very different convexity properties. 

\nin (a) In view of 
\[  \la {\cal O}_{(1,0)}, (t,x)\ra 
= t + \Theta_\omega(G)(x), \quad \mbox{ we have }\quad 
B({\cal O}_{(1,0)}) = \R \times B(\Theta_\omega(G)) \subeq \hat\g, \]
where, for a subset $C \subeq \g'$, we put 
\[  B(C) := \{ x \in \g \: \inf\la C, x \ra > - \infty\}.\] 

Therefore ${\cal O}_{(1,0)}$ is (semi-)equicontinuous in $\hat\g'$ 
if and only if $\Theta_\omega(G) \subeq \g'$ has the corresponding 
property.   

\nin (b) The subspace 
$${\cal O}_{(1,0)}^\bot  = \{ (t,x) \in \hat\g \: \Theta_\omega(G)(x) = \{-t\}\}$$
is an invariant ideal of $\hat\g$ which intersects the the central 
line $\R \times \{0\}$ trivially. 
Its image in $\g$ consists of all those 
elements $x \in \g$ for which the evaluation function 
$$ \ev_x \: \g' \to \R, \quad \ev_x(\alpha) = \alpha(x) $$
is constant on $\Theta_\omega(G)$. Since $\Theta_\omega$ is a 
$1$-cocycle with values in $\g'$, we have for any $g \in G$: 
\begin{align*}
T_g(\Theta_\omega)(g\cdot y)(x) 
&= \big(\Ad^*(g)T_e(\Theta_\omega)(y)\big)(x)
= T_e(\Theta_\omega)(y)(\Ad(g)^{-1}x)= \omega\big(\Ad(g)^{-1}x,y\big). 
\end{align*}
All these expressions vanish if and only if 
$\Ad(G)x$ is contained in the radical $\rad(\omega) 
= \{ y \in \g \: i_y \omega=0\}$. Therefore 
the image of ${\cal O}_{(1,0)}^\bot$ in $\g$ 
is the largest $\Ad(G)$-invariant ideal of $\g$ contained in $\rad(\omega)$. 
The central extension splits on this ideal. 
\end{remark}

\end{document}